\documentclass[a4paper]{article}
\usepackage{graphicx} 
\usepackage[left=2cm,right=2cm]{geometry}
\usepackage{authblk}
\usepackage{amsthm}
\usepackage{thmtools}
\usepackage{mathtools}
\usepackage{amsmath}
\usepackage{amssymb}
\usepackage{bm}
\usepackage{placeins}

\usepackage{algorithm}
\usepackage[noend]{algpseudocode}
\usepackage{amsthm}
\usepackage{hyperref}
\usepackage{mathtools}
\usepackage{microtype}
\usepackage{physics}
\usepackage{tabularx}
\usepackage{thm-restate}
\usepackage{thmtools}
\usepackage{bm}
\usepackage[capitalise]{cleveref} 
\usepackage[dvipsnames]{xcolor}
\usepackage{tikz}
\usetikzlibrary{positioning}

\usepackage[style=numeric, sorting=none]{biblatex}
\newcommand{\mathsc}[1]{{\normalfont\textsc{#1}}}

\definecolor{DarkRed}{HTML}{8B0000}

 \newtheorem{theorem}{Theorem}
 \newtheorem{definition}[theorem]{Definition}
 \newtheorem{lemma}[theorem]{Lemma}
 \newtheorem{proposition}[theorem]{Proposition}
 \newtheorem{corollary}[theorem]{Corollary}
  
 \newtheorem{problem}[theorem]{Problem}
\newtheorem{maintheorem}{Theorem}

\newcounter{protocol}

\crefname{protocol}{protocol}{protocols}
\Crefname{protocol}{Protocol}{Protocols}

\title{The Hidden Subgroup Problem in Semidirect Products and Quasi-Hamiltonian Groups}
\author[1,2]{Mauro E.S. Morales}
\affil[1]{Joint Center for Quantum Information and Computer Science (QuICS), University of Maryland, USA}
\affil[2]{Diraq, Sydney, New South Wales, Australia}
\date{}

\begin{document}

\maketitle
\begin{abstract}
Several early quantum algorithms, including Simon's algorithm and
Shor's period-finding, can be formulated as instances of
the Abelian hidden subgroup problem  (HSP).
The non-Abelian case is of
particular interest because some instances, such as
the dihedral and symmetric group HSPs, are connected to lattice
problems and graph isomorphism, respectively.
Although no efficient
quantum algorithm is known for the general non-Abelian HSP,
polynomial-time algorithms have been obtained for several special
families, including certain semidirect product groups and
Hamiltonian groups.

In this work, we give polynomial-time quantum algorithms for two
further families containing non-Abelian groups.
First, we consider groups of
the form $G=A\rtimes_{\varphi} \mathbb{Z}_{p^k}$, with  $A$ finite Abelian, $p$ prime, $k\in \mathbb{N}$ and the action of
$\mathbb Z_{p^k}$ is generated by the scalar automorphism
$a\mapsto\mu a$, for some
$\mu\in\mathbb Z_{\operatorname{Exp}(A)}^\times$, where $\mathrm{Exp}(A)$ is the exponent of $A$. Our algorithm is efficient when $A$ has bounded generator rank and $\mathrm{Exp}(A)/p=\mathrm{polylog}(|G|)$.
This includes the case $A=\mathbb{Z}_N$ for $N\in\mathbb{N}$ and $k=1$, studied by Bacon, Childs and van Dam~\cite{Bacon2005optimalmeasurement}, and $A=\mathbb{Z}_{q^r}$ with $q$ prime and $r\in \mathbb{N}$ studied by van Dam and Dey~\cite{VanDam2014semidirect}. 
Second, we give a polynomial-time quantum algorithm for finite
quasi-Hamiltonian groups under a mild assumption on the input structure. Quasi-Hamiltonian groups are finite
nilpotent groups with modular subgroup lattice, or equivalently the
finite groups in which every subgroup is permutable. As far as we know, this is the first quantum algorithm to exploit the modularity of the subgroup lattice for solving the HSP. This extends, under the aforementioned structured input assumption, the quantum algorithm for Dedekind groups given by Hallgren, Russell, and Ta-Shma~\cite{Hallgren2003HSP}.
\end{abstract}
\section{Introduction}
The hidden subgroup problem serves as a unifying framework for several problems for which quantum computers have an advantage when compared to their classical counterparts \cite{Childs2010review}.
Many early quantum algorithms such as the Deutsch-Jozsa algorithm, Simon's algorithm, and Shor's period finding algorithm \cite{Jozsa2001shor} can be formulated as instances of the hidden subgroup problem (HSP). 
In the HSP, a group $G$ is described by an encoding of its generators and access is given to an oracle $f:G\to S$, where $S$ is a set of labels. The function $f$ hides a subgroup $H\leq G$, i.e., $f$ is constant on each left coset of $H$ and evaluates to a different label on each left coset. The task is to find a list of generators of $H$. No polynomial-time classical algorithm is known for the general
non-Abelian HSP. By contrast, when $G$ is Abelian, there is an efficient quantum algorithm for the HSP \cite{Childs2010review, cheung2001decomposing}. 

No polynomial-time quantum algorithm is known for the general
non-Abelian hidden subgroup problem. The interest in the non-Abelian case lies in the fact that an efficient solution to the dihedral HSP would yield efficient
quantum algorithms for certain lattice problems that underlie
lattice-based cryptography \cite{Regev2004lattice} and the graph isomorphism problem admits an HSP formulation over
symmetric groups.
While no efficient algorithm is known for these two cases, there is a quantum subexponential algorithm for the dihedral HSP \cite{Kuperberg2005dihedral}.
The quantum query complexity of the non-Abelian HSP is known to be polynomial \cite{Ettinger2004query}. In the dihedral case, Ettinger and Høyer showed that $O(\log N)$ oracle queries together with efficiently implementable single-copy Fourier-sampling measurements produce enough classical information to determine the hidden subgroup with high probability. However, no efficient procedure is known for recovering the subgroup from these measurement outcomes \cite{Ettinger2000Dihedral}.
Separately, the optimal measurement for identifying a hidden reflection from $O(\log N)$ dihedral coset states is a joint pretty good measurement, whose efficient implementation is related to average-case subset sum and is not currently known \cite{BaconChildsVanDam2006}.

There is a long list of works where efficient quantum algorithms have been proposed for particular non-Abelian groups.
Due to the semidirect product structure of the dihedral group $D_N=\mathbb{Z}_N \rtimes \mathbb{Z}_2$, several works have studied groups with similar form. We mention some of these works in the following. In \cite{Friedl2014hiddentranslation}, the case $\mathbb{Z}_p^k \rtimes_\varphi\mathbb{Z}_2$ with $p$ prime and $\varphi_x(a)=(-1)^x a$ with $x\in \mathbb{Z}_2$ and $\mathbb{Z}_p^k$ was studied. In \cite{Moore2007power}  the authors give an efficient algorithm for the $q$-hedral case, i.e., $\mathbb{Z}_{p}\rtimes \mathbb{Z}_q$ with $p$ prime, $q\mid (p-1)$, provided that $q=(p-1)/\mathrm{polylog}(p)$ and $q$ is prime. This result was extended in \cite{Bacon2005optimalmeasurement}, where an efficient algorithm is provided for $\mathbb{Z}_N\rtimes \mathbb{Z}_p$. Other extensions include \cite{Goncalves2009metacyclic} which solves the groups $\mathbb{Z}_p\rtimes\mathbb{Z}_{q^s}$ for $p,q$ odd primes, $q\mid (p-1)$, $p/q=\mathrm{polylog}(p)$ and $s$ a positive integer, and  also \cite{VanDam2014semidirect} which solves the case $\mathbb{Z}_{p^r}\rtimes \mathbb{Z}_{q^s}$ where $p,q$ are distinct primes and such that $p^r/q^{t-j}\in\mathrm{polylog}(p^r)$ for certain group-dependent parameters $t,j$. Note that the previously mentioned works have a restriction on the relative size of the groups in the semidirect product. In \cite{Inui2007semimdirect}, the authors exploit the structure of a family of groups of the form $\mathbb{Z}_{p^r}\rtimes \mathbb{Z}_p$ with $p$ an odd prime to give a polynomial-time quantum algorithm. Notably, this algorithm does not have the condition which restricts the relative size of the groups in the semidirect product. 
In another direction \cite{Hallgren2003HSP}, it was shown that the HSP can be solved efficiently for Dedekind groups (groups that have only normal subgroups). 
For a more complete review of works on the non-Abelian HSP see \cite{dutto2025surveyhiddensubgroupproblem}.
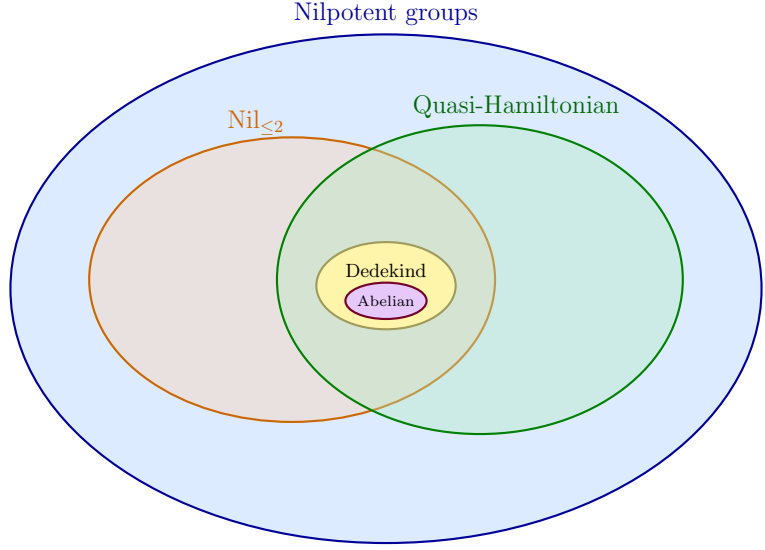
\begin{figure}[!t]
\centering

\begin{tikzpicture}[scale=0.8, transform shape]

  \definecolor{nilpotentColor}{RGB}{220,235,255} 
  \definecolor{nilTwoColor}{RGB}{255,210,170}    
  \definecolor{quasiColor}{RGB}{180,235,190}     
  \definecolor{dedekindColor}{RGB}{255,245,170}  
  \definecolor{AbelianColor}{RGB}{230,200,250}   

  \draw[
    fill=nilpotentColor,
    draw=blue!60!black,
    thick
  ]
    (0,0) ellipse (6.2cm and 4.2cm);

  \node[
    text=blue!60!black,
    font=\large
  ] at (0,4.55)
    {Nilpotent groups};

  %
  \draw[
    fill=nilTwoColor,
    fill opacity=0.38,
    draw=orange!80!black,
    draw opacity=1,
    thick
  ]
    (-1.55,0.15) ellipse (3.35cm and 2.35cm);

  \node[
    text=orange!80!black,
    font=\large
  ] at (-2.15,2.75)
    {$\mathrm{Nil}_{\leq 2}$};

  %
  \draw[
    fill=quasiColor,
    fill opacity=0.38,
    draw=green!50!black,
    draw opacity=1,
    thick
  ]
    (1.55,0.15) ellipse (3.35cm and 2.55cm);

  \node[
    text=green!40!black,
    font=\large
  ] at (2.15,3.00)
    {Quasi-Hamiltonian};

  %
  \draw[
    fill=dedekindColor,
    draw=yellow!55!black,
    thick
  ]
    (0,0.05) ellipse (1.15cm and 0.72cm);

  \node[
    font=\small
  ] at (0,0.3)
    {Dedekind};

  \draw[
    fill=AbelianColor,
    draw=purple!60!black,
    thick
  ]
    (0,-0.20) ellipse (0.67cm and 0.30cm);

  \node[
    font=\scriptsize
  ] at (0,-0.20)
    {Abelian};

\end{tikzpicture}

\caption{Schematic inclusion relations among several classes of groups. The class $\mathrm{Nil}_{\leq 2}$ includes the groups of nilpotency class $\leq 2$. For these groups, the algorithm of~\cite{Ivanyos2012Nil2} solves the HSP
efficiently. The nilpotency class of quasi-Hamiltonian groups is not bounded by a constant and therefore these two classes are incomparable (see \cref{sec:comparison}). Both classes contain the finite Dedekind groups, in which every
subgroup is normal. Note that 
the dihedral group $D_N$ is nilpotent if and only if $N$ is a power of two.}
\label{fig:group-classes}
\end{figure}

In this work, we study the hidden subgroup problem in two non-Abelian settings.
First, we give a polynomial-time quantum algorithm for groups of the form $A \rtimes \mathbb Z_{p^k}$, with $A$ Abelian, $p$ prime, $k$ positive integer, and under the condition that the generator rank of $A$ is constant and the exponent of $A$ is not much larger than $p$.
These groups contain several semidirect product families previously
studied in the HSP literature. In particular it includes the groups $\mathbb{Z}_N\rtimes \mathbb{Z}_p$ of \cite{Bacon2005optimalmeasurement} and the groups $\mathbb{Z}_{p^r}\rtimes \mathbb{Z}_{q^s}$ from \cite{VanDam2014semidirect}.
 Our algorithm is based on a reduction to the Hidden Multiple Shift problem from \cite{ivanyos2018learninglinearfunctionssubset}.
We also give an efficient quantum algorithm based on the pretty good measurement technique for the groups $\mathbb{Z}_N\rtimes \mathbb{Z}_{p^k}$, with $N\geq 1$.

Second, we give a polynomial-time quantum algorithm for finite quasi-Hamiltonian groups, under a structured presentation of the
ambient group.
 A subgroup $H\leq G$ is called \emph{permutable},
or \emph{quasinormal}, if $HK=KH$ for every subgroup $K\leq G$,
and a group is quasi-Hamiltonian if all of its subgroups are
permutable.
Every normal subgroup is permutable, although the
converse need not hold.
Quasi-Hamiltonian groups
strictly extend the class of groups in which every subgroup is
normal, which are known as Dedekind groups. 
The HSP for Dedekind groups was previously studied by
Hallgren, Russell and Ta-Shma~\cite{Hallgren2003HSP}.

Quasi-Hamiltonian groups can also be characterized by their subgroup lattice \cite{schmidt1994subgroup}. The subgroup lattice of a group $G$ consists of all subgroups of $G$, ordered by inclusion. Here, “lattice” is used in the order-theoretic sense and should not be confused with the lattices that arise in lattice-based cryptography. For finite groups, quasi-Hamiltonian groups are precisely the nilpotent groups that have a \emph{modular} subgroup lattice. 
Informally, modularity means that the subgroups fit together in a particularly regular way, analogous to the subspaces of a vector space.
To the best of our knowledge, this is the first HSP algorithm which explicitly exploits modularity of the subgroup lattice.

A natural question at this point is why care about quantum algorithm for these non-Abelian groups. Solving the HSP for many of the non-Abelian groups studied in the literature does not have practical applications (as far as we are aware). Nonetheless, there are a few reasons why studying them is important. As already mentioned, there are at least two non-Abelian groups where there are interesting applications, the dihedral and symmetric groups. Understanding what groups can be solved efficiently helps us understand what properties can be exploited by quantum algorithms and may also shed light on their limitations for the HSP.

One such property is nilpotency, which has been the focus of a substantial line of work on the HSP \cite{Ivanyos2012Nil2,Imran2024}. In \cite{Ivanyos2012Nil2} an efficient quantum algorithm was given for nilpotency class $2$ groups, while \cite{Imran2024} gave a quantum algorithm for groups with constant nilpotency class which is efficient if the prime factors of the order are bounded by a constant. If there was an efficient solution to the HSP for nilpotent groups, we would be able to solve the HSP for $D_N$ when $N$ is a power of two. A heuristic algorithm for such case was already given in \cite{Fujita2022Heuristic}. We view our result on quasi-Hamiltonian groups as complementary to the results of \cite{Ivanyos2012Nil2,Imran2024}: quasi-Hamiltonian groups are nilpotent and have modular subgroup lattices, but their nilpotency class need not be bounded by a constant.

Semidirect product groups are also relevant beyond the dihedral setting. In particular, HSP formulations over semidirect products arise in connection with the code-equivalence problem \cite{dutto2025surveyhiddensubgroupproblem}. A better understanding of the HSP over semidirect products may therefore lead to efficient algorithms for further group families in which such applications arise.

In \cref{sec:comparison} we compare our results with various works in the literature. \cref{app:comparison-semidirect} gives a comparison of our results with the work of \cite{Inui2007semimdirect}. Our results on quasi-Hamiltonian groups apply to a more general family of groups compared to their work. Nonetheless, their results apply in the black-box group model while our results require some extra data as input. In \cref{app:comparison-near-hamiltonian} it is shown that the class of quasi-Hamiltonian groups is incomparable with that of poly-near-Hamiltonian groups in \cite{Gavinsky2004nearHam}. In \cref{app:nilpotent-comparison} it is shown that the set of quasi-Hamiltonian groups is incomparable with the set of groups studied in \cite{Ivanyos2012Nil2,Imran2024} (see \cref{fig:group-classes}).
An interesting question that our work leaves open is whether the HSP can be solved efficiently when the subgroup is promised to be permutable. This would subsume the result of \cite{Hallgren2003HSP} where the normal subgroup case is studied. Note that there are subgroups in the dihedral group that are not permutable, so solving this question would not imply a solution to the dihedral HSP.

\paragraph{Statement on AI usage.} We have used ChatGPT 5.6 Sol to proofread some of the proofs in the paper and in the preparation of \cref{fig:group-classes}. A gap in the proof of \cref{lem:efficient-mu-inverse} was found by the AI and fixed.  The examples in \cref{sec:comparison} were also supplied by the AI. It also gave advice on improving the presentation, but the original idea of this work and the proofs are of human origin. The author assumes full responsibility for the correctness of the manuscript. 

\subsection{Main results}
We now explain in more detail the main results of this work. In every case, the hidden subgroup problem is specified by oracle access to a function that is constant and distinct on the left cosets of a hidden subgroup $H$. Recovering $H$ means outputting a set of generators for $H$.
We assume the group operations, evaluation of the corresponding group actions and the hiding oracle are available as reversible circuits. In the case of group operations we assume the access is as in the black-box group setting explained in \cref{sec:blackbox-group}.

 Our algorithms for
semidirect products depend on a reduction which generalizes the one in \cite{Bacon2005optimalmeasurement}.  Given a subgroup $H\leq A\rtimes_{\varphi}\mathbb
Z_{p^k}$, for $A$ Abelian, we first recover $H_A:=H\cap (A\times\{0\})$
by solving an Abelian HSP over $A$.  Whenever \(H_A\) is normal, we
factor it out.  In the resulting quotient, the residual hidden
subgroup has trivial intersection with the Abelian normal factor, and
its projection onto $\mathbb Z_{p^k}$ is therefore injective.
Consequently, the residual hidden subgroup is either trivial or
cyclic of the form $\langle(d,p^t)\rangle$
for some element \(d\) of an Abelian quotient of \(A\) and some
\(t\in\{0,\ldots,k-1\}\).  For a general action, if \(H_A\) is not
normal, the hidden subgroup is contained in
\(A\rtimes\langle p\rangle\); iterating this observation eventually
reduces the problem to the same cyclic form.  Thus, the remaining
task is to recover $t$ and $d$. 

We first treat the case $A=\mathbb{Z}_N$ using a pretty good measurement construction.  Although this group family is
also covered by \cref{thm:main-scalar}, the proof is independent and extends the
PGM approach of~\cite{Bacon2005optimalmeasurement} from a prime-order
complement to a prime-power complement. To obtain an efficient algorithm, we require the condition $N/p=\mathrm{polylog(N)}$ so that the probability of success of the PGM is not negligible. 
\begin{maintheorem}
\label{thm:main-cyclic}
Let $p$ be a prime, $k\geq1$, $N\in \mathbb{N}$ and let
$\varphi:\mathbb Z_{p^k}\to\operatorname{Aut}(\mathbb Z_N)$. If $\frac{N}{p}=\operatorname{polylog}(N)$,
then the hidden subgroup problem in $\mathbb Z_N\rtimes_{\varphi}\mathbb Z_{p^k}$
can be solved by a polynomial-time quantum algorithm.
\end{maintheorem}
Next, we consider the more general family $G=A\rtimes_{\varphi}\mathbb Z_{p^k}$,
where $A$ is a finite Abelian group.  We assume that $A$ has constant generator rank (minimum number of generators) and $\operatorname{Exp}(A)/p=\operatorname{poly}(\log \abs{G})$ where $\operatorname{Exp}(A)$ is the exponent of $A$. We assume that $\varphi$ corresponds to what we define as a \emph{scalar action}. We can always decompose $A$ into its invariant-factor form as $A=\mathbb{Z}_{N_1}\oplus \cdots \oplus \mathbb{Z}_{N_\rho}$ with $N_1\mid N_2 \mid \cdots \mid N_\rho$. In this form, an element of $A$ can be written as $a=(a_1,\cdots,a_\rho)$. We call $\varphi$ a scalar action if 
$$
\varphi_b(a)=\mu^b a=(\mu^b a_1,\cdots, \mu^b a_\rho),
$$
for $b\in \mathbb{Z}_{p^k}$ and $\mu \in \mathbb{Z}_{\operatorname{Exp}(A)}^\times$.
We use then apply the same reduction used for \cref{thm:main-cyclic}. At this point the reader might think that we could once again apply the PGM technique to obtain $d$ and $t$ as in the previous case. Unfortunately, it is not known whether this is possible for the group $A$ with the properties we described before.
Instead, we use the algorithm from \cite{ivanyos2018learninglinearfunctionssubset} for the Hidden Multiple Shift (HMS) problem.
The fact that $\varphi$ is a scalar action has two roles.  First, it makes every subgroup of $A$
invariant under the action, so $H\cap(A\times\{0\})$ is automatically normal.
Second, after taking the quotient, the restrictions
of the hiding function fulfills the multiple shift structure required by the HMS algorithm.

 The HMS problem is
formulated over groups of the form $\mathbb Z_q^n$.  Its algorithm
works in arbitrary dimension $n$, but it is polynomial-time in our
parameter regime when $n$ is constant and the number of available
shifts is sufficiently large compared with $q$.
After recovering $H_A=H\cap(A\times\{0\})=A_1\times\{0\}$,
we pass to the quotient $Q:=A/A_1$.
Then, we can write its invariant-factor decomposition as
$$
Q\cong
\mathbb Z_{q_1}\oplus\cdots\oplus\mathbb Z_{q_m}.$$
Although $Q$ need not be of the form $\mathbb Z_q^m$, it can be
embedded efficiently into $\mathbb Z_{E_Q}^m$ where $E_Q:=\operatorname{Exp}(Q)$.
The image of this embedding may be a proper subgroup, so we attach a
label specifying the coset of the image in which an element lies.
This allows us to extend the hiding function to all of
$\mathbb Z_{E_Q}^m$ while keeping it injective.
The scalar-action assumption is crucial in this reduction.  It
ensures that the same scalar acts on every cyclic component of \(Q\),
and hence that scalar shifts are preserved by the embedding. 

It remains open whether the explicit PGM construction used for the
cyclic normal factor can be extended efficiently to the general
bounded-rank abelian normal factors considered here. To reiterate our point, we remark that \cref{thm:main-cyclic} is a special case of \cref{thm:main-scalar}. We have included it as a separate result because we prove it in a different way. In \cref{thm:main-cyclic} we use the PGM approach and in \cref{thm:main-scalar} we use the HMS algorithm.
\begin{maintheorem}
\label{thm:main-scalar} Let $G=A\rtimes_\varphi \mathbb{Z}_{p^k}$ with $A$ a finite Abelian group of bounded generator rank, and $\varphi:\mathbb{Z}_{p^k}\longrightarrow\operatorname{Aut}(A)$ a scalar action. If $\operatorname{Exp}(A)/p=\operatorname{poly}(\log \abs{G})$, where $\operatorname{Exp}(A)$ is the exponent of $A$, then the hidden subgroup problem in $A\rtimes_{\varphi}\mathbb Z_{p^k}$
can be solved in quantum polynomial time.
\end{maintheorem}
For the next result, we show that there is an efficient quantum algorithm for the  quasi-Hamiltonian HSP.
A quasi-Hamiltonian group $G$ can be decomposed into its Sylow factors
\begin{align}\label{eq:G-decomp-main}
    G=\prod_{p\mid |G|\text{, $p$ prime}}G_p.
\end{align}
A key result of quasi-Hamiltonian groups is that their Sylow $p$-subgroups are modular and in fact have a rigid characterization, as shown by \cref{prop:p-group-modular}. Each Sylow $p$-subgroup $P$ must be either (i) of the form $P=Q_8\times E$ with $Q_8$ the quaternion group and $E$ is an elementary Abelian $2$-group or, (ii) $P=A\langle b \rangle$ with $A\unlhd P$ and $P/A$ a cyclic factor group. The element $b\in P$ is such that there is a positive integer $s$ such that $b$ acts on  $A$ as $b^{-1}ab=a^{1+p^s}$ and $s\geq 2$ for $p=2$. Case (i) was solved in \cite{Hallgren2003HSP} as it corresponds to the case that the group is Hamiltonian (non-Abelian Dedekind group). Case (ii) includes Abelian $p$-groups but also non-Abelian groups which have no efficient solution to the HSP in previous literature. Therefore, solving the quasi-Hamiltonian case reduces to solving the non-Abelian case (ii). 

In the proof of \cref{thm:main-quasihamiltonian}, we assume that the
Sylow decomposition in \cref{eq:G-decomp-main} and the structural
description of each factor are supplied as part of the input.
In
particular, for every non-Hamiltonian factor, we are given the data
$$
P=A\langle b\rangle
$$
together with coordinates for $A$ and the action of $b$.  We call this a structured
quasi-Hamiltonian presentation.
This is stronger than arbitrary
black-box access to \(G\): we do not address the separate
constructive-recognition problem of recovering this presentation
from an arbitrary generating set.

The crucial property of Case (ii) is the special way in which $b$
acts on $A$. Let $r:=1+p^s$,
conjugation by $b$ acts as the power automorphism $a\mapsto a^r$.
Since $r$ is relatively prime to the exponent of $A$, this is an
automorphism.  More importantly, every subgroup $K\leq A$ is fixed by the action of $b$ since $K^r=K$.

We exploit this property by constructing an associated Abelian
$p$-group \(B\). Informally, $B$ is essentially group $P$, but the element corresponding to $b$ is made
to commute with $A$. We can relate the elements of $B$ and $P$ through a bijection $\sigma:B\to P $. The map $\sigma$ is not a group isomorphism.  It is
instead a \emph{crossed isomorphism}: its failure to preserve
multiplication is controlled by a family of twisting automorphisms.
Concretely, it satisfies an identity of the form
$$
\sigma(x)\sigma(y)
=
\sigma\bigl(f(y)(x)y\bigr),
$$
where $f:B\to \mathrm{Aut}(B)$, with $\mathrm{Aut}(G)$ the automorphisms of $G$. The maps $f(y)$ are powers of the automorphism
$x\mapsto x^r$.

The crossed isomorphism property by itself does not imply that
subgroups and cosets are preserved.  The additional point in our
construction is that every twisting map $f(y)$ is a power
automorphism of the Abelian $p$-group $B$, and hence fixes every
subgroup of $B$. 
It can then be shown by Baer's theorem \cite{Baer1944CrossedIsomorphisms} (see also \cref{thm:schmidt:baers-258}) that $\sigma$ induces an isomorphism between the subgroup lattices and preserves right
cosets.  In particular,
$\sigma$ maps subgroups of $B$ bijectively to subgroups of $P$
and preserves right cosets:
$$
\sigma(Kx)=\sigma(K)\sigma(x)
\qquad
\text{for every }K\leq B\text{ and }x\in B.
$$
This property allows us to transport the HSP from the non-Abelian group
\(P\) to the Abelian group \(B\). We can therefore apply the standard Abelian HSP algorithm to recover the generators of the hidden subgroup in $B$. 
\begin{maintheorem}[Structured quasi-Hamiltonian groups]
\label{thm:main-quasihamiltonian}
Let $G$ be a finite quasi-Hamiltonian group given by a structured
quasi-Hamiltonian presentation as above. Then, given an oracle
hiding an arbitrary subgroup $H\leq G$, there is an efficient quantum
algorithm that outputs generators of $H$.
\end{maintheorem}

\section{Preliminaries}
\subsection{Group theory preliminaries}\label{sec:group-prelim}
In this section we give a summary of some of the concepts in group theory that are used in this work. If $X$ is a subset of a group $G$, then $\langle X\rangle$ is the generated subgroup of $G$. The \emph{exponent} of a finite group $G$, denoted $\mathrm{Exp}(G)$ is the least common multiple of the order of the elements of $G$. An \emph{elementary Abelian group} is a group isomorphic to $\mathbb{Z}_{p}^n$ for some prime $p$ and natural number $n$.

Let $G, K$ be two groups. The \emph{semidirect product} $G\rtimes_{\varphi} K$ is a group defined in terms of a homomorphism $\varphi:K\to \mathrm{Aut}(G)$. For simplicity, we write the automorphism $\varphi(k)\in \mathrm{Aut}(G)$ as $\varphi_k$. The underlying set of the semidirect product is the Cartesian product $G\times K$ and the group operation is given by $(g,k)\cdot(g',k')=(g\cdot \varphi_k(g'),k\cdot k')$. This defines a group with identity $(e_G,e_K)$ and inverse
$$(g,k)^{-1}=(\varphi_{k^{-1}}(g^{-1}),k^{-1})=(\varphi_k^{-1}(g^{-1}),k^{-1}).$$

A relevant group for us is $A\rtimes \mathbb{Z}_{p^k}$ with $A$ an Abelian group. Note that since $\mathbb{Z}_{p^k}$ is cyclic, $\varphi_b = \varphi_1 \circ \cdots \varphi_1 = \varphi^{b}$ where $b\in \mathbb{Z}_{p^k}$ and we abuse notation by defining $\varphi = \varphi_{1}$. 

Let $p$ be a prime. A finite \emph{$p$-group} $G$ is a group for which every element has order a power of $p$. Equivalently, $\abs{G}=p^k$ for some $k\geq 0$. If $H\leq G$ and $H$ is $p$-group, we may call $H$ a $p$-subgroup. A \emph{Sylow $p$-subgroup} of $G$ is a maximal $p$-subgroup of $G$. 

\subsubsection{Solvable and nilpotent groups}
Both solvable and nilpotent groups are defined in terms of nested series of normal subgroups. This series is called the subnormal series.
\begin{definition}
    A subnormal series of a group $G$ is a sequence of subgroups $\{G_i\}_{i=0}^n$ where each is a normal subgroup of the next and $G_0=\{e\}, G_n=G$, i.e., 
    \begin{align}
        \{e\}=G_0\unlhd G_1 \unlhd \cdots \unlhd G_n = G.
    \end{align}
    The quotients $G_{i+1}/G_i$ are called the factors.
\end{definition}
\begin{definition}
    A composition series is a subnormal series such that each factor $G_{i+1}/G_i$ is simple (it has no normal subgroups except itself and the trivial group).
\end{definition}
Intuitively, this is saying that one cannot break up the subnormal series anymore because the factors are simple.
\begin{definition}
    A finite group $G$ is solvable if it has a subnormal series whose factors are all Abelian.
\end{definition}
Note that this means that the composition series of a solvable finite group $G$ has factors that are cyclic of prime order (the only finite Abelian simple groups are the cyclic groups of prime order.)
\begin{definition}
    Given a group $G$, define the commutator $[G,G]$ as the subgroup generated by all commutators $[g,h]=g^{-1}h^{-1}gh$ for $g,h\in G$.
\end{definition}
If $G$ is Abelian, then $[G,G]=\{e\}$. In a sense, $[G,G]$ is measuring how far the group is from being Abelian. We have that $G/[G,G]$ is Abelian, because the non-commuting part of $gh=hg[g,h]$ has been factored out.
\begin{definition}
    The derived series of $G$ is a sequence of groups $\{G^{(i)}\}_{i=0}$ such that $G^{(i+1)}=[G^{(i)},G^{(i)}]$ and $G^{(0)}=G$.
\end{definition}
\begin{proposition}
A group $G$ is solvable if and only if there exists an integer
$n\geq0$ such that $G^{(n)}=\{e\}$.
\end{proposition}
\begin{definition}
    Let $\gamma_1(G)=G$ and $\gamma_{i+1}(G)=[G,\gamma_i(G)]$. A group $G$ is nilpotent if $\gamma_n(G)=\{e\}$ for some $n$.
\end{definition}
The least $c$ such that $\gamma_{c+1}(G)=\{e\}$ is called the nilpotency class of $G$.

\subsubsection{Permutability, quasi-Hamiltonian groups and lattices}
Below, we explain how permutability and lattice theory are related. We begin by defining permutability, then explaining the basics of lattice theory required to understand the proof in this work, and then present some results that relate both concepts.
Several of the definitions and results presented in this section can be found in Chapter 2 of \cite{schmidt1994subgroup} (Sections 2.1 to 2.5).  Our starting point is the notion of a permutable subgroup.
\begin{definition}
    Let $G$ be a group. A subgroup $H\leq G$ is permutable if for all $K\leq G$, $HK=KH$.
\end{definition}
Permutability generalizes the notion of a normal subgroup. When a group has only permutable subgroups, the group is called \emph{quasi-Hamiltonian}. These groups can be characterized through their subgroup lattices, which are defined in the paragraphs below.

\paragraph{Basics of Lattices}A partially ordered set (poset) is a tuple $(S,\leq )$ consisting of a set $S$ and a partial order $\leq$. 
A partial order $\leq$ is a relation on a set $S$ such that for all $a,b,c\in S$ one has (1) $a\leq a$, (2) $a\leq b$ and $b\leq a$ implies $a=b$ and (3) $a\leq b$ and $b\leq c$ implies $a\leq c$.
A lattice is a tuple $(L,\wedge,\vee)$ consisting of a set $L$ and two binary commutative and associative operations $\wedge$ and $\vee$ (called meet and join) satisfying the following identities: for all $a,b\in L$, $a\vee (a\wedge b)=a$ and $a\wedge (a \vee b)=a$. In a lattice one can define the relation $\leq$ as follows: $a\leq b $ if and only if $a=a\wedge b$. It is easy to check that $(L,\leq)$ is a poset. For simplicity, sometimes we abuse notation and call $L$ the lattice.

An example of a lattice is the subset lattice. Given some set $X$, the set of all subsets of $X$ is a poset with respect to set inclusion. The operator $\vee$ corresponds to set union and $\wedge$ to set intersection.
In this work we focus on subgroup lattices. Let $G$ be a group, define the subgroup lattice $L(G)$ as the set of all subgroups, the operation $X \wedge  Y$ is set intersection $X\cap Y$ and $X\vee Y$ corresponds to the join of subgroups $\langle X,Y \rangle$, i.e., the generated subgroup by the union of $X$ and $Y$. The partial order is given by subgroup inclusion.

The map $\sigma:L\to \tilde{L}$ is a lattice homomorphism if $\sigma(x\wedge y)=\sigma(x)\wedge \sigma(y)$ and $\sigma(x\vee y)=\sigma(x)\vee \sigma(y)$. Following the literature, a lattice isomorphism from $L(G)$ to $L(\widetilde{G})$ for groups $G,\tilde{G}$ is called a \emph{projectivity}. 
In particular, suppose that $\sigma:G\longrightarrow\widetilde G$
is a bijection that maps subgroups of $G$ bijectively to subgroups
of $\widetilde G$. Then the map $K\longmapsto \sigma(K)$
induces a projectivity from $L(G)$ to $L(\widetilde G)$. When no
confusion can arise, we use the same symbol $\sigma$ for the
element map and the induced map on subgroup lattices.
Crossed isomorphisms, which provide the projectivity used in our
quantum algorithm for quasi-Hamiltonian groups, are introduced in \cref{sec:poly-quasiHam}
immediately before their application.

An important property of a lattice is that of modularity. A lattice $L$ is modular if for all $x,y,z\in L$, the following holds: If $x\leq z$, then $x\vee (y\wedge z)=(x\vee y)\wedge z$.
An element $m$ of a lattice $L$ is \emph{modular} in $L$ if (1) $x\vee (m \wedge z)= (x\vee m) \wedge z$ for all $x,z\in L$ and $x\leq z$, (2) $m\vee (y\wedge z)=(m\vee y)\wedge z$ for all $y,z\in L$ with $m\leq z$.
A finite group $G$ is called modular if its subgroup
lattice $L(G)$ is modular. When $G$ is a $p$-group, we also say
that $G$ is a \emph{modular $p$-group}.
\paragraph{Permutability and lattices} It can be shown that permutability is connected to the concept of a modular subgroup lattice. 
\begin{theorem}[Theorem 2.1.3 in \cite{schmidt1994subgroup}]\label{thm:Schmidt213}
    Let $G$ be a group. If $K\leq G$ such that $KH=HK$ for all $H\leq G $, then $K$ is modular in $L(G)$.
\end{theorem}
\begin{lemma}[Lemma 2.3.2 in \cite{schmidt1994subgroup}]\label{lem:Schmidt232}
    A finite $p$-group has a modular subgroup lattice if and only if any two subgroups permute.
\end{lemma}
In this work, we focus on quasi-Hamiltonian groups, which correspond to groups where all subgroups are permutable.
\begin{definition}
    A finite quasi-Hamiltonian group is a group $G$ such that all subgroups are permutable.
\end{definition}
Now we prove \cref{lem:quasiHam-nil} which is useful in understanding the structure of quasi-Hamiltonian groups. This result is a combination of a few results in \cite{schmidt1994subgroup} combined with an exercise in the same book. For this reason we include a proof the lemma.
\begin{lemma}\label{lem:quasiHam-nil}
    A finite group $G$ is quasi-Hamiltonian group if and only if $G$ is nilpotent and has a modular subgroup lattice.
\end{lemma}
\begin{proof}
    We first prove the $\Rightarrow$ direction. Since every $H\leq G$ is permutable, by \cref{thm:Schmidt213} we have that $H$ is a modular element, and since this is true of every element in the lattice of $G$, then $G$ is modular. To show nilpotency we use the fact that $G$ is nilpotent if and only if every Sylow $p$-subgroup is normal. Let $P\leq G$ be a Sylow $p$-group of $G$ and define $P^{g}=gPg^{-1}$ for $g\in G$, which is also a Sylow $p$-subgroup. Since subgroups permute we have $PP^g = P^g P$ and moreover $PP^g$ is a subgroup with order $\abs{PP^g}=\frac{\abs{P}\abs{P^g}}{\abs{P\cap P^g}}$. Note that $\abs{P^g}$ is a power of $p$ and $P\cap P^g$ is a subgroup of $P$ and therefore $\abs{P\cap P^g}$ is also a power of $p$. Then, $PP^g$ is a $p$-subgroup, which implies that $PP^g\leq P$.
    Since we also have $P\leq PP^g$, then $P^g\leq P$. Both $P$ and $P^g$ have the same order and thus $P=P^g$. 

    Now we prove the $\Leftarrow$ direction. Since $G$ nilpotent, it can be decomposed in the direct product of its Sylow subgroups.
    \begin{align}\label{eq:G-nil}
        G\cong \prod_{p\mid \abs{G}} G_p, 
    \end{align}
    with each $G_p$ normal.
    Since $L(G)$ is modular then for each $G_p$ we have that $L(G_p)$ is modular. By \cref{lem:Schmidt232}, any two subgroups of $G_p$ permute. Let $H,K\leq G$ and consider its Sylow decomposition $H_p=H\cap G_p$ and $K_p = K\cap G_p$. Then,
    \begin{align}
        HK=(\prod_p H_p)(\prod_p K_p)&=\prod_p H_p K_p\\
        &=\prod_p K_p H_p\\
        &=KH.
    \end{align}
\end{proof}
\begin{proposition}\label{prop:Iwasawa-sylow-decomp}
    A finite group $G$ is quasi-Hamiltonian if and only if
    \[
    G\cong \prod_{p\mid \abs{G}} G_p,
    \]
    where the $G_p$ are the Sylow subgroups of $G$ and each $G_p$ is a modular $p$-group.
\end{proposition}
\begin{proof}
    $(\Rightarrow)$. By \cref{lem:quasiHam-nil}, $G$ is nilpotent and can be decomposed as in \cref{eq:G-nil}. Since it also has a modular subgroup lattice the implication follows.

    $(\Leftarrow)$. Let $H, K\leq G$ and consider $H_p$ and $K_p$ as in the proof of \cref{lem:quasiHam-nil}. Then $H_p$ and $K_p$ permute and by a similar argument $H$ and $K$ permute.
\end{proof}

Thus, finite quasi-Hamiltonian groups decompose into modular Sylow $p$-groups, and the HSP on quasi-Hamiltonian groups reduces to understanding those factors.

\begin{proposition}[Theorem 2.3.1 in \cite{schmidt1994subgroup}, see also \cite{Ballester2009permutable}]\label{prop:p-group-modular}
    A finite $p$-group $P$ is modular if and only if one of the following holds:
    \begin{enumerate}
        \item $P\cong Q_8 \times E$, where $Q_8$ is the quaternion group of order $8$ and $E$ is an elementary Abelian $2$-group.
        \item $P$ has an Abelian normal subgroup $A$ with cyclic factor group $P/A$; furthermore there exists $b\in P$ with
        $$
        P= A  \langle b\rangle
        $$
        and a positive integer $s$ such that $b$ acts on $A$ as $b^{-1}ab=a^{1+p^s}$ with $s\geq 2$ for $p=2$.
    \end{enumerate}
\end{proposition}
Note that case 2 in \cref{prop:p-group-modular} includes the case of an Abelian $p$-group.
\subsubsection{Black-box groups}\label{sec:blackbox-group}
The results in this work require that the groups be represented in the model of black-box groups \cite{Babai1984matrix}. In this model, a finite group $G$ is given by encodings of its elements as bit strings of a fixed length, together with oracles implementing the group operations. Thus, given encodings of $g,h\in G$, the oracle returns an encoding of $gh$, and given an encoding of $g$, it returns an encoding of $g^{-1}$. The input also contains a list of encoded generators $(g_1,\ldots,g_k)$, with the promise that $G=\langle g_1,\ldots,g_k\rangle$.
If $N\triangleleft G$ is a known normal subgroup, the quotient $G/N$ can also be regarded as a black-box group. An element of $G/N$ is a coset $gN$, and we encode this coset by any representative $g\in G$. Multiplication and inversion in the quotient are implemented by multiplication and inversion of representatives in $G$: $(gN)(hN)=ghN$ and $(gN)^{-1}=g^{-1}N$. Note that as described above, the encoding is non-unique (elements of the group can be encoded in more than one way).

Some of the algorithms used as subroutines in this work require a unique encoding, particularly those from \cite{Friedl2014hiddentranslation}. As described in Theorem 4.8 of \cite{Friedl2014hiddentranslation}, if a group $G$ is solvable and $N$ is normal such that $G/N$ is Abelian, then with $O(\log \abs{G})$ copies of $\ket{N}$ it is possible to obtain a unique encoding for $G/N$. To prepare the state $\ket{N}$, Theorem 4.3 of \cite{Friedl2014hiddentranslation} can be used which says that if $G$ has a unique encoding and a subgroup $H\leq G$ is solvable, then $\ket{H}$ can be prepared efficiently with negligible error. These facts will be enough for the cases we cover in our work.

Let $\ell$ denote the encoding length of a group element.  We
measure the running time in $\ell$, and assume that $\ell=\operatorname{poly}(\log |G|)$. The black-box oracles can also be implemented as unitaries \cite{Watrous2000succint}. In this work we assume access to such oracles.
\section{Reduction for groups $A\rtimes \mathbb{Z}_{p^k}$}\label{sec:reduction-general}

In this section, we reduce the hidden subgroup problem in $G=A\rtimes \mathbb{Z}_{p^k}$ to a simpler problem based on a generalized version of the reduction by Ettinger and H{\o}yer \cite{Ettinger2000Dihedral}.
Let $f:G\to S$ hide a subgroup $H\leq G$ and $\varphi:\mathbb Z_{p^k}\to\operatorname{Aut}(A)$ the action in the group $A\rtimes_{\varphi} \mathbb{Z}_{p^k}$ as described in \cref{sec:group-prelim}.
Following \cite{Bacon2005optimalmeasurement}, we first solve the HSP on the Abelian group
$$
G_1=A\times\{0\}.
$$
The restriction of $f$ to $G_1$ hides the subgroup $H_1:=H\cap (A\times \{0\})=A_1\times \{0\}$.
Since this is an Abelian HSP, we can recover generators of $H_1$ efficiently. 
The next step is to determine whether $H_1$ is normal in $G$. Because the ambient group $G$ is solvable, we can use the algorithm from \cite{watrous2001quantumalgorithmssolvablegroups} to check whether $H_1$ is normal in $G$.
\begin{lemma}\label{lem:semidirect-solvable}
    Let $G=A\rtimes_{\varphi} C$ with $A$ Abelian and $C$ cyclic. Then $G$ is solvable.
\end{lemma}
\begin{proof}
    The subgroup $A\times\{e\}$ is a normal Abelian subgroup of $G$. Hence
    $$
    \{(e,e)\}\unlhd A\times\{e\}\unlhd G,
    $$
    and the quotient $G/(A\times\{e\})\cong C$ is Abelian. Therefore $G$ is solvable.
\end{proof}
The following lemma is a useful characterization for the case that $H_1$ is normal in $G$.
\begin{lemma}\label{lem:norm-phi}
    $H_1\unlhd G$ if and only if $\varphi(A_1)=A_1$.
\end{lemma}
\begin{proof}
    Let $(a,0)\in H_1$ and $(x,y)\in G$. Then
    $$
    (x,y)(a,0)(x,y)^{-1}=(\varphi^{y}(a),0).
    $$
    If $H_1$ is normal, then $\varphi^y(a)\in A_1$ for all $y\in \mathbb{Z}_{p^k}$, and in particular $\varphi(a)\in A_1$. Hence $\varphi(A_1)\subseteq A_1$. Since $\varphi$ is an automorphism, equality follows.

    Conversely, if $\varphi(A_1)=A_1$, then $\varphi^{y}(a)\in A_1$ for every $y$, so the conjugate of any $(a,0)\in H_1$ again lies in $H_1$. Thus $H_1\unlhd G$.
\end{proof}
Next, we consider the cases when $H_1$ is not normal in $G$. If this is the case, the hidden subgroup $H$ can be shown to be inside a smaller semidirect product group.
\begin{proposition}\label{prop:not-norm}
    If $H_1$ is not normal in $G=A\rtimes \mathbb{Z}_{p^k}$, then
    $$
    H\leq K_1,
    \qquad
    K_1=A\rtimes \langle p \rangle \cong A\rtimes \mathbb{Z}_{p^{k-1}}.
    $$
\end{proposition}
\begin{proof}
    Assume $H\not\leq K_1$. Then $H$ contains some element $(a,c)$ with $a\in A$ and $c\in\mathbb{Z}_{p^k}$ such that $p\nmid c$. Hence, $c$ is invertible modulo $p^k$, so there exists $m\in\mathbb{Z}_{p^k}$ with $mc\equiv 1\pmod{p^k}$. Therefore, $(a,c)^m\in H$ has second coordinate $1$, so $(d,1)\in H$ for some $d\in A$.

    Now let $(h,0)\in H_1$. Then
    $$
    (d,1)(h,0)(d,1)^{-1}=(\varphi(h),0)\in H.
    $$
    From this we conclude that $\varphi(h)\in A_1$, so $\varphi(A_1)=A_1$. By \cref{lem:norm-phi} this implies $H_1\unlhd G$, contradicting the assumption. Therefore $H\leq K_1$.
\end{proof}
It follows that if $H_1$ is not normal, we then proceed to solve the hidden subgroup problem in the group $A\rtimes \mathbb{Z}_{p^{k-1}}$. We may then repeat the previous analysis and check whether $H_1$ is normal in $A\rtimes \mathbb{Z}_{p^{k-1}}$. For $0\leq j\leq k$, define
$$
G_j
:=
A\rtimes_{\phi}\langle p^j\rangle
=
\left\{
(a,p^j b)
:
a\in A,\;
b\in\mathbb Z_{p^{k-j}}
\right\}.
$$
Note that we can identify $A\rtimes_{\phi^{p^j}}\mathbb Z_{p^{k-j}}$ with $G_j$. Under the map $\iota_j:
A\rtimes_{\phi^{p^j}}\mathbb Z_{p^{k-j}}
\to G_j$, such that $\iota_j(a,b)=(a,p^j b)$,
the group $G_j$ is a semidirect product whose cyclic factor has
order \(p^{k-j}\).

Suppose inductively that \(H\leq G_j\). By Lemma~\ref{lem:norm-phi},
$$
H_1\trianglelefteq G_j
\quad\Longleftrightarrow\quad
\phi^{p^j}(A_1)=A_1.
$$
If $H_1$ is not normal in $G_j$, Proposition~\ref{prop:not-norm}
applied inside $G_j$ gives $H\leq G_{j+1}$.
Thus the reduction proceeds until either $H_1\trianglelefteq G_j$
for some $j<k$, or $j=k$. In the latter case,
$$
H\leq G_k=A\times\{0\},
$$
and hence $H=H_1$.
\subsection{The case $H_1\unlhd G$}
We proceed to analyze the case when $H_1$ is normal. In this case, we show that the problem can be reduced to the hidden subgroup problem on a semidirect product group where the hidden subgroup is either cyclic or trivial. 
\begin{lemma}\label{lem:quotient-cases}
    Let $G=A\rtimes_{\varphi} \mathbb{Z}_{p^k}$ and assume $H_1\unlhd G$. Write
    $$
    H_1=H\cap (A\times \{0\})=A_1\times \{0\},
    $$
    and set
    $$
    G_2=G/H_1\cong A_2\rtimes_{\varphi} \mathbb{Z}_{p^k},\qquad A_2:=A/A_1,
    $$
    and $H_2:=H/H_1$. Then $H_2$ is either trivial or of the form $H_2=\langle (d,p^t)\rangle$
    for some $t\in\{0,1,\dots,k-1\}$ and some $d\in A_2$.
\end{lemma}
\begin{proof}
    If $H=H_1$, then $H_2$ is trivial. Assume now that $H\neq H_1$.

    Consider the projection
    $$
    \pi:G_2\to \mathbb{Z}_{p^k},\qquad \pi((a,b))=b.
    $$
    Its kernel is $A_2\times\{0\}$. Restricting $\pi$ to $H_2$, we have
    $$
    \ker(\pi|_{H_2})=H_2\cap (A_2\times\{0\}).
    $$
    We claim that this kernel is trivial. Indeed, if $xH_1\in H_2\cap (A_2\times\{0\})$, then $x\in H$ and also $x\in A\times\{0\}$ modulo $H_1$, so in fact $x\in H\cap (A\times\{0\})=H_1$. Therefore $xH_1$ is the trivial coset.

    Hence $\pi|_{H_2}$ is injective, so $H_2\cong \pi(H_2)\leq \mathbb{Z}_{p^k}$. Every subgroup of $\mathbb{Z}_{p^k}$ is either trivial or generated by $p^t$ for some $t\in\{0,\dots,k-1\}$. Therefore $H_2$ is either trivial or cyclic of order $p^{k-t}$, generated by some element of the form $(d,p^t)$.
\end{proof}
After taking the quotient by $H_1$, the HSP is therefore reduced to the case in which the hidden subgroup has the form $H=H_d^t:=\langle (d,p^t)\rangle$,
with trivial intersection with $A\times\{0\}$.
Note that it must be true that
\begin{align}
    \sum_{j=0}^{p^{k-t}-1}\varphi^{jp^t}(d)=0,
\end{align}
because the order of $H$ is $p^{k-t}$. This gives rise the following problem.
\begin{problem}\label{prob:reduced}
    Given $G=A\rtimes \mathbb{Z}_{p^k}$ with $A$ Abelian, and a function $f:G\to S$ hiding $H\leq G$ where either $H=\langle(d,p^t)\rangle$ for $d\in A$, $t\in\{0,\cdots,k-1\}$ or $H$ is trivial. Moreover, $H\cap (A\times\{0\})$ is trivial. Find a set of generators of $H$.
\end{problem}
Note that the groups $H_{d,t}=\langle (d,p^t) \rangle$ can also be written as 
\begin{align}
    H_{d,t}=\left\{\left(\sum_{j=0}^{b-1}\varphi^{j p^t}(d),b p^t\right) : b\in \mathbb{Z}_{p^{k-t}}\right\}.
\end{align}
Since $\mathrm{Aut}(\mathbb{Z}_N)\cong \mathbb{Z}_{N}^{\times}$, an automorphism $\varphi$ corresponds to multiplying by some $\mu \in \mathbb{Z}_N^{\times}$. Therefore, we have $\varphi^{jp^t}(d)=\mu^{jp^t} d$. We can then define 
\begin{align}\label{eq:M}
    M_t^{(b)}=\sum_{j=0}^{b-1}\mu^{j p^t},
\end{align}
and thus,
\begin{align}
    H_{d,t}=\{(M_t^{(b)}d,bp^t): b\in Z_{p^{k-t}}\}.
\end{align}
\subsection{Polynomial-time algorithm for $\mathbb{Z}_{N}\rtimes \mathbb{Z}_{p^k}$}\label{sec:poly-ZN}

Our starting point is \cref{prob:reduced}.
We first consider the case in
which $t$ is known. We will first show how the pretty good measurement strategy in \cite{Bacon2005optimalmeasurement} can also be used in this problem. Then, we show that even if $t\in\{0,\cdots,k-1\}$ is not known, we can use a linear search strategy to solve the problem. 

A fact that will be useful is that the complete set of left coset representatives is given by $L_t=\{(\ell,c): \ell\in \mathbb{Z}_N, c\in \mathbb{Z}_{p^t}\}$.
\begin{proposition}
    Let $H_{d,t}\leq \mathbb{Z}_N \rtimes \mathbb{Z}_{p^k}$. Then $L_t=\{(\ell,c): \ell\in \mathbb{Z}_N, c\in \mathbb{Z}_{p^t}\}$ is a complete set of left coset representatives.
\end{proposition}
\begin{proof}
    First, we show that given a $(a,x)\in \mathbb{Z}_N \rtimes \mathbb{Z}_{p^k}$, there is some $\ell\in \mathbb{Z}_N$, $c\in \mathbb{Z}_{p^t}$ such that $(a,x)\in (\ell,c)H_{d,t}$. Since $x\in \mathbb{Z}_{p^k}$, we can always write it as $x=c+rp^t$ for some $c\in\{0,\cdots,p^t-1\}$ and $r\in \mathbb{Z}_{p^{k-t}}$. Then, note that for any $\ell$, 
    \begin{align}
        (\ell,c)(M^{(r)}_{t}d,rp^t)&=(\ell+\mu^{c }M_t^{(r)}d,c+rp^t).
    \end{align}
    Then, pick $\ell=a-\mu^c M^{(r)}_t d  \quad (\text{mod } N)$. Indeed, this shows that $(a,x)$ is in some left coset of the form described above.

    To show uniqueness, suppose that $(\ell,c)H_{d,t}=(\ell',c') H_{d,t}$. Then, we have that
    \begin{align}\label{eq:cosets-equal}
        (\ell,c)(M^{(r)}_t d, rp^t)=(\ell',c')(M^{(r')}_t d, r' p^t),
    \end{align}
    for some $r,r'\in \mathbb{Z}_{p^{k-t}}$. Comparing the second coordinate of each side gives $c+rp^t = c'+r'p^t \pmod{p^k}$. Since $c,c'\in\{0,\cdots,p^t-1\}$, we can consider the equality modulo $p^t$, giving $c=c'$ which implies $r p^t= r' p^t \pmod{p^k}$ or $r=r' \pmod{p^{k-t}}$. Since $r,r'\in \mathbb{Z}_{p^{k-t}}$, then $r=r'$. From \cref{eq:cosets-equal} we get $\ell=\ell'$.
\end{proof}
\subsubsection{PGM for fixed $t$}\label{sec:fixed-t}
 In this section, we give a description of how to solve \cref{prob:reduced} in the fixed $t$ case, i.e., we assume $t$ is known and fixed. The pretty good measurement approach consists in preparing the quantum states encoding the hidden subgroup and then the pretty good measurement is used to distinguish between the possible hidden subgroups. Since $t$ is fixed, this section adapts the method described in \cite{Bacon2005optimalmeasurement}.

First, we write the left coset states $\ket{\psi_{\ell,c,d,t}}:=\ket{(\ell,c) H_{d,t}}$ as 
\begin{align}
\ket{\psi_{\ell,c,d,t}}&=\frac{1}{\sqrt{p^{k-t}}}\sum_{b\in \mathbb{Z}_{p^{k-t}}} \ket{(\ell, c)(M_t^{(b)}d,bp^t)}\\
&= \frac{1}{\sqrt{p^{k-t}}}\sum_{b\in \mathbb{Z}_{p^{k-t}}} \ket{\ell +\mu^c M_t^{(b)}d,c+bp^t}.
\end{align}
Suppose that the access to $f$ is given by the oracle $U_f$. We can use this oracle to prepare the state  
\begin{align}
    \frac{1}{\sqrt{N p^k}}\sum_{(\ell,c)\in A\rtimes \mathbb{Z}_{p^k}}\ket{(\ell,c)}\ket{0}\to \frac{1}{\sqrt{N p^k}}\sum_{(\ell,c)\in A\rtimes \mathbb{Z}_{p^k}}\ket{(\ell,c)}\ket{f(\ell,c)}.
\end{align}
By discarding the second register, the first register can be described by a mixed state as follows
\begin{align}
    \rho_{d,t}=\frac{1}{Np^t}\sum_{\ell\in \mathbb{Z}_N}\sum_{c\in\mathbb{Z}_{p^t}}\ketbra{\psi_{\ell,c,d,t}}.
\end{align}
Next, we show that this coset state can be put in a convenient form to apply the PGM, a similar strategy was followed in \cite{Bacon2005optimalmeasurement}. Define the unitary $W_t$ as acting in the following way on the register holding the state $\ket{c+bp^t}$ together with an ancilla:
\begin{align}
W_t\ket{c+bp^t}\ket{0}=\ket{c}\ket{b}.
\end{align}
Note that this is a unitary because the operation is reversible as $c\in \mathbb{Z}_{p^t}$ and $b\in\mathbb{Z}_{p^{k-t}}$. The purpose of this unitary is to extract the coset label $c$ and $b$ from the second register. Applying $W_t$ to the coset state gives
\begin{align}
W_t\ket{\psi_{\ell,c,d,t}}\ket{0}=\frac{1}{\sqrt{p^{k-t}}}\sum_{b\in\mathbb{Z}_{p^{k-t}}}\ket{\ell + \mu^c M_t^{(b)}d,c,b}.
\end{align}
If we act on the coset state with both the Fourier transform of the group $\mathbb{Z}_N$ and $W_t$ we get
\begin{align}
    \mathcal{F}_N\otimes W_t\ket{\psi_{\ell,c,d,t}}\ket{0}=\frac{1}{\sqrt{N}}\sum_{x\in \mathbb{Z}_N} \omega_N^{x\ell}\ket{x,c}\otimes \ket{\phi_{x,c,d,t}}
\end{align}
where $\ket{\phi_{x,c,d,t}}:=\frac{1}{\sqrt{p^{k-t}}}\sum_{b\in\mathbb{Z}_{p^{k-t}}}\omega_N^{x\mu^{c}M_t^{(b)}d}\ket{b}$.

We can average over $\ell$, giving the final (transformed) coset state as
\begin{align}
    \widetilde{\rho}_{d,t}=(\mathcal{F}_N\otimes W_t)\rho_{d,t}(\mathcal{F}_N\otimes W_t)^\dagger=\frac{1}{Np^t}\sum_{x\in\mathbb{Z}_N}\sum_{c\in\mathbb{Z}_{p^t}}\ketbra{x,c}\otimes \ketbra{\phi_{x,c,d,t}}.
\end{align}
We can further simplify the state by noticing that we can remove the label $c$. Define the unitary 
\begin{align}
    U_t\ket{x,c}=\ket{x\mu^c,c},
\end{align}
which is well defined because $\mu^c\in \mathbb{Z}_N^\times$. Note that $x$ and $c$ appear in $\ket{\phi_{x,c,d,t}}$ as $y:=x\mu^c$ and therefore this unitary allows us to write the two first registers as a maximally mixed state, which has no useful information.
\begin{align}
    (U_t\otimes I)\widetilde{\rho}_{d,t}(U_t\otimes I)^\dagger&=\frac{1}{Np^t}\sum_{y\in\mathbb{Z}_N}\sum_{c\in\mathbb{Z}_{p^t}}\ketbra{y,c}\otimes \ketbra{\phi_{y,d,t}},\\
    &= \frac{I_{p^t}}{p^t}\otimes \sigma_{d,t}
\end{align}
with $\ket{\phi_{y,d,t}}=\frac{1}{\sqrt{p^{k-t}}}\sum_{b\in \mathbb{Z}_{p^{k-t}}}\omega_N^{yM_t^{(b)}d}\ket{b}$ and 
\begin{align}\label{eq:sigma-hidden-state}
    \sigma_{d,t}=\frac{1}{N}\sum_{y\in \mathbb{Z}_N}\ketbra{y}\otimes \ketbra{\phi_{y,d,t}}.
\end{align}
For fixed \(t\), it remains to recover the unknown parameter
\(d\in\mathbb Z_N\) from the states \(\sigma_{d,t}\).  In \cref{sec:app-PGM}, we analyze the exact one-copy PGM for the ensemble $\{\sigma_{d,t}:d\in\mathbb Z_N\}$
and prove the lower bound
$$
P_{\mathrm{succ}}^{(t,1)}
\geq
\frac{\phi(N)p}{N^2}.
$$
In \cref{sec:app-PGM} we also describe an efficiently implementable block
measurement, which achieves the same lower
bound.  We use this block measurement as the fixed-$t$ subroutine
in the algorithm below.
\subsubsection{Finding $d$ and $t$}\label{sec:d-t}
In this section we assume that both $d$ and $t$ are unknown. We show that using the procedure in \cref{sec:fixed-t} together with a linear search over the possible values of $t$ is enough to solve this problem. 
\begin{lemma}\label{lem:verification}
    Let $H=H_{d,t}=\langle (d,p^t) \rangle$. Then:  
    \begin{itemize}
        \item if $s<t$, there is no element of $H$ whose second component is $p^s$.
        \item if $s\geq t$, there is exactly one element of $H$ whose second component is $p^s$, namely $h_s=(M_t^{(p^{s-t})}d,p^s)$.
    \end{itemize}
\end{lemma}
\begin{proof}
    Any element of $H$ has the form $(M_t^{(m)}d,mp^t)$ for $m\in \mathbb{Z}_{p^{k-t}}$. If $s<t$, then the second coordinate can't be $p^s$ since $m p^t$ is always divisible by $p^t$ but $p^s$ is not. If $s\geq t$, then we want $mp^t=p^s \pmod{p^k}$, which implies $m=p^{s-t} \pmod{p^{k-t}}$. This choice of $m$ gives the unique element $h_s$. Note that when $s=t$ we have that $h_t=(d,p^t)$.
\end{proof}
Based on this lemma, we construct a test to detect if for a given $s$ we have that $s\geq t$. Define for $s\in \{0,\cdots, k-1\}$ and $c\in \mathbb{Z}_N$,
\begin{align}
    V(s,c) = \begin{cases}
    1 & \text{if } f(0,0)=f(c,p^s), \\
    0  & \text{otherwise. }  
\end{cases}
\end{align}
For the hiding oracle $f:G\to S$, we have that $f(e)=f(g) \iff g\in H$, where $e=(0,0)$. Therefore, we have that $V(s,c)=1 \iff (c,p^s)\in H$. For a given $s$ and $c$, an algorithm can compute $V(s,c)$ by doing the following:
\begin{itemize}
    \item Query oracle $f$ on $(0,0)$ and get $f(0,0)$.
    \item Query the oracle on $(c,p^s)$ and obtain $f(c,p^s)$.
    \item Compare if $f(0,0)$ and $f(c,p^s)$ are equal. If they are, output $1$, otherwise output $0$. 
\end{itemize}
We now describe the full algorithm. Let
$\epsilon\in(0,1)$ be the desired error probability. Denote by
$\mathcal{P}_s$ the fixed-$s$ procedure from \cref{sec:fixed-t}. Thus, when the hidden subgroup is of the form $H_{c,s}=\langle(c,p^s)\rangle$,
the procedure $\mathcal{P}_s$ outputs \(c\) with probability at least $q_0:=\frac{\varphi(N)p}{N^2}$.
Indeed, \cref{sec:app-PGM} gives this as a lower bound for the success probability of the PGM.
If the action of $\langle p^s\rangle$ on $\mathbb Z_N$ is
trivial, the corresponding fixed-$s$ problem is Abelian and can
instead be solved using the Abelian HSP algorithm with success
probability at least $q_0$.
Set
\begin{align}
R:=
\left\lceil
\frac{1}{q_0}\log\left(\frac{k}{\epsilon}\right)
\right\rceil
=
\left\lceil
\frac{N^2}{\varphi(N)p}
\log\left(\frac{k}{\epsilon}\right)
\right\rceil .
\end{align}

The algorithm is as follows.
\begin{enumerate}
    \item Query the oracle once on the identity and store
    $f(0,0)$.

    \item For each $s\in\{0,1,\ldots,k-1\}$, in increasing
    order, repeat the following procedure $R$ times:
    \begin{enumerate}
        \item Prepare a fresh coset state and run the fixed-$s$
        procedure $\mathcal{P}_s$, obtaining a classical output
        $c\in\mathbb Z_N$.

        \item Compute $V(s,c)$, equivalently compare
        $f(c,p^s)$ with the previously stored value $f(0,0)$.

        \item If $V(s,c)=1$, output $H=\langle(c,p^s)\rangle$ and halt.
    \end{enumerate}
    \item If none of the tests accepts, output$H=\{(0,0)\}$.
\end{enumerate}
The verification step is important. Although the fixed-$s$
procedure is guaranteed to recover the hidden parameter only when
$s$ is the correct value, the test \(V(s,c)\) prevents the
algorithm from accepting an element which does not belong to the
hidden subgroup.
\begin{proposition}
The algorithm above outputs a generating set for the hidden
subgroup with probability at least $1-\epsilon$.
\end{proposition}
\begin{proof}
First suppose that the hidden subgroup is trivial. For every
$s\in\{0,\ldots,k-1\}$ and every $c\in\mathbb Z_N$, the element
$(c,p^s)$ is nontrivial because $p^s\neq 0\pmod {p^k}$.
Consequently, $V(s,c)=0$
for every candidate produced by the algorithm. The algorithm
therefore reaches the last step and outputs
$H=\{(0,0)\}$. Thus, in the trivial-subgroup case, the output is
correct with probability one.

Now suppose that the hidden subgroup is $H=H_{d,t}=\langle(d,p^t)\rangle$ for some $t\in\{0,\ldots,k-1\}$. Consider first a stage
$s<t$. By \cref{lem:verification}, there is no element of \(H\) whose second
coordinate is $p^s$. Hence
$$
V(s,c)=0
$$
for every $c\in\mathbb Z_N$. It follows that the algorithm cannot
halt at a stage $s<t$, regardless of the output of
$\mathcal{P}_s$.

At the stage $s=t$, Lemma~21 shows that the unique element of
$H$ whose second coordinate is $p^t$ is $h_t=(d,p^t)$. 
Therefore,
$$
V(t,c)=1
\quad\Longleftrightarrow\quad
c=d.
$$
Each independent execution of $\mathcal{P}_t$, using a fresh
coset state, outputs $d$ with probability at least $q_0$.
Thus, the probability that none of the $R$ executions produces
$d$ is at most
$$
(1-q_0)^R
\leq
e^{-q_0R}
\leq
\frac{\epsilon}{k}
\leq
\epsilon.
$$
Whenever one of the executions produces $d$, the verification test accepts and the algorithm outputs
$$
H=\langle(d,p^t)\rangle,
$$
which is the correct hidden subgroup.

For completeness, suppose that all $R$ executions at $s=t$
fail and the algorithm proceeds to some $s>t$. By \cref{lem:verification}, the
unique element of $H$ whose second coordinate is $p^s$ is
$$
h_s
=
\left(
M_t^{(p^{s-t})}d,p^s
\right)
=
(d,p^t)^{p^{s-t}}.
$$
Thus, if a candidate is accepted at stage $s$, the algorithm
outputs the subgroup
$$
\langle h_s\rangle
=
\left\langle(d,p^t)^{p^{s-t}}\right\rangle.
$$
Since $(d,p^t)$ has order $p^{k-t}$, the element $h_s$ has
order
$$
p^{k-s}<p^{k-t}.
$$
Therefore, $\langle h_s\rangle<H$
whenever $s>t$. Such an output would be incorrect. However, the
algorithm can reach a stage $s>t$ only if all $R$ executions at
the true stage $t$ have failed. Hence the probability of
outputting such a proper subgroup is at most
\[
(1-q_0)^R\leq\frac{\epsilon}{k}.
\]
Similarly, if the algorithm reaches the end and incorrectly
outputs the trivial subgroup, then it must also have missed \(d\)
in every execution at stage \(t\). This event has the same
probability bound.

Consequently, every possible incorrect output is contained in the
event that all \(R\) fixed-\(t\) executions fail. Therefore,
\[
\Pr[\text{incorrect output}]
\leq
(1-q_0)^R
\leq
\frac{\epsilon}{k}
\leq
\epsilon,
\]
which proves the claim.
\end{proof}
We finally analyze the running time. The algorithm makes at most
\(kR\) calls to a fixed-\(s\) procedure and at most \(kR+1\)
queries for the verification tests. Since
\[
\frac{1}{q_0}
=
\frac{N^2}{\varphi(N)p}
=
\frac{N}{p}\frac{N}{\varphi(N)}
\]
and
\[
\frac{N}{\varphi(N)}=O(\log\log N),
\]
we have
\[
R
=
O\left(
\frac{N}{p}\log\log N
\log\left(\frac{k}{\epsilon}\right)
\right).
\]
Under the assumption
$$
\frac{N}{p}=\operatorname{poly}(\log N),
$$
the number of repetitions is polynomial in
$\log N$ and $\log(k/\epsilon)$. Therefore, since the
fixed-\(s\) procedure is efficient, the complete linear-search
algorithm runs in quantum time polynomial in $\log N, k\log p$, and $\log(1/\epsilon)$.

The value $\varphi(N)$ need not be supplied as part of the
input. It may be computed efficiently after factoring $N$ with
Shor's algorithm. Alternatively, one may choose $R$ using a
standard lower bound for $\varphi(N)/N$.

\subsection{The hidden subgroup problem with scalar action}
\label{sec:scalar-action-hsp}
In this section we consider the hidden subgroup problem in $G=A\rtimes_\varphi \mathbb{Z}_{p^k}$, where $A$ is a finite Abelian group and the action of $\mathbb{Z}_{p^k}$ on $A$ corresponds to scalar multiplication (in a sense defined below). We show that HSP in this family reduces to the Hidden Multiple Shifts (HMS) problem of \cite{ivanyos2018learninglinearfunctionssubset}. We give a descriptions of the HMS \cref{sec:HMS-explanation}. The resulting algorithm is efficient when the generator rank of $A$ is bounded and the exponent of $A$ is not much larger than $p$. The algorithm is correct for every finite Abelian group $A$.

\subsubsection{The hidden multiple shifts problem}\label{sec:HMS-explanation}
We give briefly review the Hidden Multiple Shift (HMS) problem as described in \cite{ivanyos2018learninglinearfunctionssubset}.
The input to this problem is an oracle $f_s:\mathbb{Z}_q^n \times R \to S$, for some $s\in \mathbb{Z}_q^n$, $R\subseteq \mathbb{Z}_q$ with $q>1$ an integer and $S$ is a set of labels (like in the HSP). This oracle has the property that $f_s(x,h)=f(x-hs)$ for $x\in\mathbb{Z}_q^n$ and $h\in R$ and for some injective function $f:\mathbb{Z}_q^n\to S$. The set $R$ can be thought of as labeling the shifts of the function $f$. The output of the HMS problem is the shift $s\in \mathbb{Z}_q^n$.

As noted in \cite{ivanyos2018learninglinearfunctionssubset}, when $q$ is a composite number, there is an ambiguity in the solution to the HMS problem, due to the fact that the oracle only shows the shifted functions indexed by $h\in R$. If all those $h\in R$ are congruent modulo some divisor of $q$, then changing $s$ in certain specific directions just translates the unknown function $f$ such that the oracle cannot detect the change. 

For completeness, we give more details regarding this ambiguity in the solution of the HMS problem, but the original explanation can be found in \cite{ivanyos2018learninglinearfunctionssubset}. Let $\delta(q,R)$ be the largest divisor of $q$ such that $\delta(q,R)$ divides $h-h'$ for all $h,h'\in R$, i.e., define it as
$$
\delta(q,R)
:=
\gcd\bigl(q,\{h-h_0:h\in R\}\bigr).
$$
Equivalently, if $h_0\in R$ is fixed, then $\delta=\delta(q,R)$ is the largest divisor of $q$ such that $h=h_0 \pmod{\delta}$ for all $h\in R$.  Let $s'\in\frac{N}{\delta}\mathbb{Z}^n_q$ and write $s'=\frac{q}{\delta}y$. Moreover, $h-h_0=\delta a_h$ for some $a_h\in \mathbb{Z}_q$. Then, we have that $(h-h_0)s'=qa_h y= 0 \pmod{q}$ for every $h\in R$. Therefore, we have that
\begin{align}
    f_{s+s'}(x,h)&= f(x-h(s+s')),\\
    &= f(x-hs-h_0 s'),\\
    &=f'(x-hs),
\end{align}
where in the last line we have defined the injective function $f'(x)=f(x-h_0 s')$. Therefore, the oracle with hidden shift $s$ and function $f'$ cannot be distinguished from the hidden shift $s+s'$ with function $f$. Given the preceding discussion, we define the HMS problem below.
\[
\begin{array}{ll}
\mathsc{Hidden Multiple Shift}(q,n,r) & \\[2mm]
\textbf{Input: }\begin{array}[t]{l}
\text{A set $R\subseteq \mathbb{Z}_q$ such that $\abs{R}=r$.}\\
\text{Oracle access to a function $f_s:\mathbb{Z}_q^n\times R \to S$, where $s\in\mathbb{Z}_q^n$, and}\\
\text{$f_s(x,h)=f(x-hs)$ for some injective $f:\mathbb{Z}_q^n\to S$.}
\end{array}
\\[1mm]
\textbf{Output: }  \text{$s \pmod{\frac{q}{\delta(q,R)}}$.}
\end{array}
\]
In \cite{ivanyos2018learninglinearfunctionssubset} a quantum algorithm to solve the HMS problem was given, and we cite the result below.
\begin{theorem}[Theorem 1.6 in \cite{ivanyos2018learninglinearfunctionssubset}]\label{thm:HMS}
There is a quantum algorithm that solves
$\operatorname{HMS}(q,n,r)$ with high probability in time
\[
O\left(
\operatorname{poly}(n)
\left(\frac{q}{r}\right)^{n+O(1)}
\right).
\]
\end{theorem}
\subsubsection{Input model and scalar actions}\label{sec:input-hms}
As usual in the HSP, the group $A\rtimes \mathbb{Z}_{p^k}$ is given as input through its generators. We assume that $A$ is given in invariant-factor form,
$$
A\cong
\mathbb Z_{N_1}\oplus\cdots\oplus\mathbb Z_{N_\rho},
\qquad
1<N_1\mid N_2\mid\cdots\mid N_\rho.
$$
The number $\rho$ is the generator rank of $A$, that is, the
minimum number of elements needed to generate $A$. The exponent of
$A$ is $E:=\operatorname{Exp}(A)=N_{\rho}.$
 Every finite Abelian group admits such a
decomposition \cite{DummitFoote}.
We assume that the cyclic generators and the
corresponding coordinate maps are part of the input. Alternatively,
they may be computed using a finite Abelian group decomposition
algorithm such as that of~\cite{cheung2001decomposing}.

Every integer $u$ acts on the additive group $A$ by scalar
multiplication, i.e., my the map $a\longmapsto ua$.
Since $Ea=0$ for every $a\in A$, this action depends only on the
residue class of $u$ modulo $E$. Moreover, multiplication by $u$ is
an automorphism of $A$ whenever $u\in\mathbb Z_E^\times$.

\begin{definition}[Scalar action]
\label{def}
An action of $\mathbb Z_{p^k}$ on $A$ is called scalar if there is an element $\mu\in\mathbb Z_E^\times$
such that
$$
\mu^{p^k}=1\pmod E
$$
and $\varphi^b(a)=\mu^b a$ for every $a\in A$ and $b\in\mathbb Z_{p^k}$.
\end{definition}
In invariant-factor coordinates, this means
$$
\varphi^b(a_1,\ldots,a_{\rho})
=
\bigl(
\mu^b a_1\bmod N_1,\ldots,
\mu^b a_{\rho}\bmod N_{\rho}
\bigr).
$$
Thus the same scalar $\mu^b$ acts on every cyclic component of
$A$.
The significance of scalarity is that every subgroup of
$A$ is invariant under the action.
\subsubsection{Reduction of an arbitrary hidden subgroup}
\label{subsec:reduction-subgroup}
Following the reduction in \cref{sec:reduction-general}, we show that solving the HSP in $A\rtimes \mathbb{Z}_{p^k}$ can be simplified by first considering an Abelian ``slice'' of the whole group.
Define $H_1:=H\cap(A\times\{0\})=A_1\times\{0\}$ and note that the group $A\times\{0\}$ with hidden subgroup $H_1$ is an instance of the Abelian HSP. Thus we can recover the generators of $H_1$ and also $A_1$.
We next show that $H_1$ is normal in $G$. Let $a\in A_1$ and
$(x,b)\in G$, then we have that
$$
(x,b)(a,0)(x,b)^{-1}
=
(\mu^b a,0).
$$
Note that $A_1$ is closed under multiplication by the integer $\mu^b$, and hence $\mu^b a\in A_1$.
Therefore, $(x,b)(a,0)(x,b)^{-1}\in A_1\times\{0\}=H_1$ and it follows that $H_1\trianglelefteq G$.
Therefore, unlike the case of a general action as in \cref{sec:reduction-general}, it is not necessary
to descend through the sequence of subgroups $A\rtimes\langle p\rangle, A\rtimes\langle p^2\rangle,\cdots$, etc.
For a scalar action, $H_1$ is normal and one may pass directly to the quotient by $H_1$. Let $Q=A/A_1$ and note that $Q$ is an Abelian group. We can consider the invariant-factor decomposition of $Q$.
\begin{align}\label{eq:Q-decomp}
Q
\cong
\mathbb Z_{q_1}\oplus\cdots\oplus\mathbb Z_{q_m},
\qquad
1<q_1\mid q_2\mid\cdots\mid q_m.
\end{align}
With $m\leq\rho$. Let $E_Q:=\operatorname{Exp}(Q)=q_m$
when $Q$ is nontrivial. The invariant-factor decomposition of $Q$ can be obtained using the algorithm in \cite{cheung2001decomposing}. This algorithm also gives efficiently implementable maps $\pi_A:A\longrightarrow Q$ and $s_Q:Q\longrightarrow A$
such that $\pi_A(s_Q(x))=x$
for every $x\in Q$. Here $\pi_A$ is the quotient homomorphism and
$s_Q$ is a canonical choice of representative.
Since $Q$ is a quotient of $A$, $E_Q\mid E$.
Scalar multiplication by $\mu$ is also induced in $Q$. Indeed, let
$$
\overline\varphi^b(a+A_1)
:=
\mu^b a+A_1.
$$
This is well defined because $\mu^bA_1=A_1$.
We continue to write $\mu$ for its residue class modulo $E_Q$. The quotient group is therefore $\overline G:=Q\rtimes_{\mu}\mathbb Z_{p^k}$.

Define $\pi:G\longrightarrow\overline G$ by
$\pi(a,b)=(\pi_A(a),b)$.
Then $\pi$ is a surjective homomorphism with $\ker(\pi)=H_1$.
Consequently,  $G/H_1\cong Q\rtimes_{\mu}\mathbb Z_{p^k}$.

An HSP oracle $\overline f:\overline G\longrightarrow S$ can be defined naturally by $\overline f(\pi(g)):=f(g)$. 
The function $\overline f$ is well defined and hides $\overline H:=\pi(H)\cong H/H_1$.
Moreover, it can be implemented using
$$
\overline f(x,b)
=
f(s_Q(x),b).
$$
Just as in \cref{lem:quotient-cases}, we have that $\overline H\cap(Q\times\{0\})=
\{(0,0)\}$.
Moreover, by \cref{lem:quotient-cases} it follows that $\overline H$ is either trivial or has the form $\overline H=\langle(d,p^t)\rangle$
for some $d\in Q$ and $t\in\{0,\ldots,k-1\}$.
In the latter case, $(d,p^t)$ has order $p^{k-t}$.
Thus the original HSP has been reduced to \cref{prob:reduced}. We solve
this problem below using the algorithm for the HMS in \cite{ivanyos2018learninglinearfunctionssubset}.

If $Q$ is trivial, then $\overline G\cong\mathbb Z_{p^k}$ is Abelian. In that case, we solve the quotient HSP using the
Abelian HSP algorithm and lift its generators to $G$. We therefore
assume throughout the remainder of the section that $Q$ is
nontrivial.

Suppose first that $t$ is known and $\overline H
=
\langle(d,p^t)\rangle$. Let $\nu_t:=\mu^{p^t}\pmod{E_Q}$
and, for $b\geq0$,
$$
M_t^{(b)}
:=
\sum_{j=0}^{b-1}\nu_t^j
\pmod{E_Q}.
$$
A direct induction using the semidirect-product multiplication
gives $
(d,p^t)^b=\bigl(M_t^{(b)}d,bp^t\bigr)$.
Here $M_t^{(b)}d$ denotes scalar multiplication in $Q$.
If $\nu_t=1\pmod{E_Q}$ then $\langle p^t\rangle$ acts trivially on $Q$. Hence $Q\rtimes_{\mu}\langle p^t\rangle
\cong
Q\times\mathbb Z_{p^{k-t}}$ is Abelian and the fixed-$t$ problem can therefore be solved directly
using the Abelian HSP algorithm.

Suppose that $\nu_t\neq1\pmod{E_Q}$. Looking towards the solving of the HMS define
\begin{align}\label{eq:R_t_def}   
R_t
:=
\left\{
M_t^{(0)},M_t^{(1)},\ldots,M_t^{(p-1)}
\right\}
\subseteq\mathbb Z_{E_Q}.
\end{align}
\begin{lemma}
\label{lem:Rt-cardinality}
If $\nu_t\neq1$, then $|R_t|=p.$
\end{lemma}
\begin{proof}
Suppose that $M_t^{(a)}=M_t^{(b)}$ for some $0\leq b<a<p$, and set $j=a-b$. Then
$$
0
=
M_t^{(a)}-M_t^{(b)}
=
\nu_t^bM_t^{(j)}.
$$
Since $\nu_t\in \mathbb{Z}_{E_Q}^\times$, this implies $M_t^{(j)}=0$.
Using the fact that $(\nu_t-1)M_t^{(j)}=\nu_t^j-1$,
we obtain $\nu_t^j=1$.
The order of $\nu_t$ is a nontrivial power of $p$, because it is
nontrivial and divides $p^{k-t}$. Therefore $p\mid j$. This is
impossible because $0<j<p$ and hence the elements of $R_t$ are pairwise distinct.
\end{proof}
The nontrivial element $\nu_t$ has $p$-power order in
$\mathbb Z_{E_Q}^{\times}$. Therefore $p\mid\varphi(E_Q)$ (where $\varphi$ here is the Euler totient function, do not confuse with the automorphism of the semidirect product), and in particular $2\leq p\leq E_Q-1$.
Thus $q=E_Q$ and $r=p$ satisfy the parameter requirements of the
Hidden Multiple Shift problem.

\subsubsection{Reduction to the Hidden Multiple Shift}

Next, we show that the HSP in $\overline G=Q\rtimes \mathbb{Z}_{p^k}$ can be mapped to an instance of the HMS problem. 

Let $f_0:Q\longrightarrow S$ such that $f_0(x):=\overline f(x,0)$.

\begin{lemma}
\label{lem:f0-injective}
The function $f_0$ is injective.
\end{lemma}
\begin{proof}
If $f_0(x)=f_0(y)$, then
$$
(x,0)^{-1}(y,0)
=
(y-x,0)
\in\overline H.
$$
Since $\overline H \cap (Q\times \{0\})=\{(0,0)\}$, this implies
$x=y$.
\end{proof}
For $0\leq b<p$, we have
$$
(x,bp^t)
=
\bigl(x-M_t^{(b)}d,0\bigr)
\bigl(M_t^{(b)}d,bp^t\bigr).
$$
The second factor belongs to $\overline H$ and thus $\overline f(x,bp^t)=f_0\bigl(x-M_t^{(b)}d\bigr)$.
We can think of $\overline f$ as having a multiple-shift structure over $Q$, where the hidden shift is $d\in Q$. Note that the factors $M_t^{(b)}\in R_t$, we $R_t$ was defined in \cref{eq:R_t_def}.
To apply the HMS
theorem exactly as stated, we convert it into an HMS instance over $\mathbb Z_{E_Q}^m$.

Recall \cref{eq:Q-decomp}. 
For each $i$, define $\lambda_i:=\frac{E_Q}{q_i}$ and let $\iota:Q\longrightarrow\mathbb Z_{E_Q}^m$ such that
$$
\iota(x_1,\ldots,x_m)
=
(\lambda_1x_1,\ldots,\lambda_mx_m).
$$
We use $\iota$ (which can be implemented efficiently) to map the elements of $Q$ to those of $\mathbb{Z}_{E_Q}^{m}$ and we prove some useful properties which we later use to show that solving the HMS in $\mathbb{Z}_{E_Q}^{m}$ also solves the problem in $Q$. 
\begin{lemma}
\label{lem:quotient-embedding}
The map $\iota$ is an injective homomorphism and satisfies $\iota(hx)=h\iota(x)$ for every integer $h$ and every $x\in Q$.
\end{lemma}

\begin{proof}
The map is a homomorphism since $\iota(x_1+y_1,\cdots,x_m+y_m)=(\lambda_1 (x_1+y_1),\cdots,\lambda_m(x_m+y_m))=(\lambda_1 x_1,\cdots\lambda_m x_m)+(\lambda_1 y_1,\cdots,\lambda_m y_m)$.
If $\iota(x)=0$, then
$\lambda_ix_i=0\pmod{E_Q}$ for every $i$. Since $E_Q=\lambda_iq_i$, this implies
$x_i=0\pmod{q_i}$.
Therefore $x=0$, and $\iota$ is injective. The identity $\iota(hx)=h\iota(x)$ can easily be seen to hold.
\end{proof}
Note that $\iota$ is injective but not surjective necessarily. 
Every element $z=(z_1,\ldots,z_m)\in\mathbb Z_{E_Q}^m$
has a unique decomposition
$$
z=c(z)+\iota(\chi(z)),
$$
where $c(z)=(c_1(z_1),\cdots,c_m(z_m))$ with $0\leq c_i(z)<\lambda_i$ and $\chi(z)\in Q$.
Explicitly, the functions $c$ and $\chi$ are defined for $0\leq z_i<E_Q$ as
\begin{align}
c_i(z):=z_i\bmod\lambda_i    
\end{align}
and
\begin{align}
\chi_i(z)
:=
\frac{z_i-c_i(z)}{\lambda_i}
\pmod{q_i}.
\end{align}
Let $D:=\iota(d)\in\mathbb Z_{E_Q}^m$ correspond to the hidden shift in $Z_{E_Q}^m$ (recall that $d$ can be understood as a hidden shift for $f_0$).
Since \cref{thm:HMS} solves the HMS over $\mathbb{Z}_{E_Q}^m$ but $f_0$ is defined on $Q$. Moreover, note that $\iota(Q)\leq \mathbb{Z}_{E_Q}^{m}$ and in general both groups are not equal. We need to define a function $g$ such that there are no collisions (i.e. the function is injective). A naive attempt would be to use the map $z\mapsto f_0(\chi(z))$, but note that this would not be an injective function. Define
$$
g:\mathbb Z_{E_Q}^m
\longrightarrow
\left(
\prod_{i=1}^m
\{0,\ldots,\lambda_i-1\}
\right)\times S
$$
by
$$
g(z)
:=
\bigl(c(z),f_0(\chi(z))\bigr).
$$
One can think of $c(z)$ as a coset label and $\iota(\chi(z))$ as labeling an element inside the coset. Therefore, $g$ has multiple copies of $f_0$ in each coset.
\begin{lemma}
\label{lem:g-injective}
The function $g$ is injective.
\end{lemma}
\begin{proof}
Suppose that $g(z)=g(z')$. Then
$c(z)=c(z')$ and $f_0(\chi(z))=f_0(\chi(z'))$.
Since $f_0$ is injective, $\chi(z)=\chi(z')$.
The uniqueness of the decomposition $z=c(z)+\iota(\chi(z))$ then gives $z=z'$.
\end{proof}
Let $
\gamma_t:\mathbb Z_p\longrightarrow R_t$ with $\gamma_t(b):=M_t^{(b)}$.
By Lemma~\ref{lem:Rt-cardinality}, $\gamma_t$ is a bijection.
For $z\in\mathbb Z_{E_Q}^m$ and $h\in R_t$, let $b:=\gamma_t^{-1}(h)$
and define
$$
\widetilde F_t(z,h)
:=
\bigl(
c(z),
\overline f(\chi(z),bp^t)
\bigr).
$$
Note that $\widetilde F$ itself can be constructed efficiently since all the functions involved can. The following proposition clarifies how $\widetilde F$ is related to the HMS problem.
\begin{proposition}
\label{prop:HMS-reduction}
For every $z\in\mathbb Z_{E_Q}^m$ and $h\in R_t$, $\widetilde F_t(z,h)=g(z-hD)$.
Thus $\widetilde F_t$ is an instance of $\operatorname{HMS}(E_Q,m,p)$ with hidden shift $D$.
\end{proposition}

\begin{proof}
Since $D=\iota(d)$ and $\iota$ commutes with scalar multiplication,
$$
z-hD
=
c(z)+\iota(\chi(z)-hd).
$$
Subtracting $hD$ does not change the residue of the $i$-th
coordinate modulo $\lambda_i$, therefore we have that $c(z-hD)=c(z)$
and also $\chi(z-hD)=\chi(z)-hd$.
It follows that
$$
g(z-hD)
=
\bigl(
c(z),
f_0(\chi(z)-hd)
\bigr).
$$

If $h=M_t^{(b)}$, then
$$
f_0(\chi(z)-hd)
=
\overline f(\chi(z),bp^t).
$$
Therefore
$$
g(z-hD)
=
\widetilde F_t(z,h).
$$
\end{proof}
In the present HMS instance there is no ambiguity in the hidden
shift. Indeed, $M_t^{(0)}=0$ and $M_t^{(1)}=1$,
so $0,1\in R_t$. Therefore, $\delta(R_t,E_Q)=1$,
because every divisor of $E_Q$ that divides all differences of
elements of $R_t$ must divide $1$.

The output of the HMS algorithm is 
\[
D\pmod{\frac{E_Q}{\delta(R_t,E_Q)}}
=
D\pmod{E_Q}.
\]
Since $D\in\mathbb Z_{E_Q}^m$, this determines $D$ exactly.
Applying Theorem~\ref{thm:HMS} for oracle $\widetilde F$ gives a fixed-$t$ algorithm with
running time
$$
\operatorname{poly}(m,\log E_Q)
\left(\frac{E_Q}{p}\right)^{m+O(1)}.
$$
Once $D$ is recovered, the element $d\in Q$ is recovered from
$D=\iota(d)$ by
$d_i=\frac{D_i}{\lambda_i}\pmod{q_i}$.

Note that it is not necessary to list the $p$ elements of $R_t$. Given
$b\in\mathbb Z_p$, the value $M_t^{(b)}$ can be computed in time
polynomial in $\log E_Q+\log p$ using repeated squaring and
$$
M_t^{(a+b)}
=
M_t^{(a)}
+
\nu_t^aM_t^{(b)}.
$$

Conversely, if $h=M_t^{(b)}$, then
$$
1+(\nu_t-1)h
=
\nu_t^b
\pmod{E_Q}.
$$
The exponent $b$ can therefore be recovered using the quantum
discrete logarithm algorithm in the cyclic subgroup
$\langle\nu_t\rangle
\leq
\mathbb Z_{E_Q}^{\times}$
Because $0\leq b<p$ and the order of $\nu_t$ is divisible by $p$,
this value of $b$ is unique.
Thus $\gamma_t$ and $\gamma_t^{-1}$ can be implemented coherently.
This permits the preparation of a uniform superposition over $R_t$
and the coherent implementation of $\widetilde F_t$.

We have shown how to solve the problem when $t$ is fixed, for the case that $t$ is unknown, we can simply follow the method in \cref{sec:d-t}.

\section{Polynomial-time algorithm for quasi-Hamiltonian groups}\label{sec:poly-quasiHam}
In this section we describe an efficient quantum algorithm for groups which have all their subgroups as permutable. As mentioned in the introduction, this is a generalization of Hamiltonian and Abelian groups. By \cref{prop:Iwasawa-sylow-decomp}, any finite quasi-Hamiltonian group can be decomposed into modular $p$-groups. There are already efficient algorithms for the HSP when $P$ is Abelian and when $P$ is Hamiltonian as in case (1) of \cref{prop:p-group-modular}. Therefore, we give an algorithm for the non-Hamiltonian case (2). Before describing the algorithm, we give a description of the input model.

\paragraph{Input model.}
Throughout this section, we assume that finite quasi-Hamiltonian group $G$
is given by a description constructed from its Sylow
decomposition. More precisely, the input is assumed to be already factorized into the Sylow factors, i.e., the group $G$ is decomposed as
$$
G=\prod_{p\mid |G|}G_p,
$$
where each $G_p$ is the Sylow $p$-subgroup of $G$. Since $G$ is quasi-Hamiltonian, each $G_p$ corresponds to one of the cases in \cref{prop:p-group-modular}. Each is described by an encoding of the generators according to the structure of the factor.
In particular, for every
non-Hamiltonian factor $P=G_p$, the input additionally specifies
the data
\[
P=A\langle b\rangle,\qquad A\trianglelefteq P,\qquad
P/A=\langle bA\rangle,
\]
together with the parameters describing the action of $b$ on
$A$ and the generators of $A$.
Moreover, we assume that it is known how $b$ acts on the group $A$. This means that in case (2) of \cref{prop:p-group-modular} we know the value of $p$ and $s$.
This is a promise in the structure of the input. It does not restrict the class of quasi-Hamiltonian groups under consideration,
but it supplies additional computational information beyond an
arbitrary black-box description of $G$. By
\cref{prop:Iwasawa-sylow-decomp}, every finite quasi-Hamiltonian group admits a
direct product decomposition into modular Sylow $p$-subgroups,
and \cref{prop:p-group-modular} gives the corresponding structure of each
factor. Since the ambient group is promised to be
quasi-Hamiltonian, we choose to represent it in a form which makes
this promised structure explicit. 

The Sylow decomposition also reduces the HSP on $G$ to the HSP
on its factors. Let $F:G\to S$ hide a subgroup $H\leq G$, and
define $H_p:=H\cap G_p$.
The restriction of $F$ to $G_p$ hides $H_p$, since for
$x,y\in G_p$,
\[
F(x)=F(y)
\quad\Longleftrightarrow\quad
x^{-1}y\in H
\quad\Longleftrightarrow\quad
x^{-1}y\in H\cap G_p.
\]
Moreover, every subgroup of the direct product of the normal
Sylow subgroups decomposes as
\[
H=\prod_{p\mid |G|}H_p.
\]
 Consequently, it is sufficient to solve the HSP separately
in each \(G_p\) and then take the union of the resulting generating
sets as detailed in \cref{sec:Sylow-HSP}.

\subsection{Solving the HSP for groups in case (ii) of \cref{prop:p-group-modular}}
The following corollary specifies the properties fulfilled by non-Hamiltonian groups in case (2), these properties will be used in the algorithm.
\begin{corollary}[Corollary 2.4.15 in \cite{schmidt1994subgroup}]\label{cor:schmidt-nonham}
    Let $G$ be a non-Hamiltonian finite $p$-group with modular subgroup lattice. Then the exponent of $G$ is $p^n$ and there is an Abelian normal subgroup $A$ of exponent $p^k$, $b\in G$ and $s\geq 1$---or $s\geq 2$ when $p=2$---such that $G/A$ is a cyclic group of order $p^m$ (with $k,m\in \mathbb{N}$) and such that $G=A\langle b \rangle$, $s<k\leq s+m$, $b^{-1}ab=a^{1+p^s}$ for all $a\in A$. Moreover, the order of $b$ is $p^n$ and $s+m=n$.
\end{corollary}
The algorithm to solve the HSP in a non-Hamiltonian group $P$ which falls in case (2) of \cref{prop:p-group-modular}, consists in mapping $P$ to an Abelian $p$-group $B$ which is defined below. As will be defined shortly, the mapping is a so-called crossed isomorphism \cite{schmidt1994subgroup}, which captures the idea that $P$ and $B$ are isomorphic up to a ``twist'' in the group multiplication. This crossed isomorphism also defines a projectivity which preserves the subgroup lattice. Below, we explain how the crossed isomorphism is constructed and then use it to solve the HSP. 

Let $P=A\langle b\rangle$ be a group such as in \cref{cor:schmidt-nonham} and set $q=p^m$ and $r=1+p^s$. Consider the HSP oracle $F:P\to S$ such that $F(u)=F(v)$ iff $uH=vH$ for some $H\leq P$. Before constructing the Abelian group $B$ and the crossed isomorphism between $B$ and $P$, we define some functions that will be needed in the construction. Define the function $\mu:\mathbb{N}\to \mathbb{Z}$ such that $\mu(j)=1+r+r^2+\cdots+r^{j-1}$. The function $\mu$ is used to define an automorphism on the group $A$. For this purpose, the following number theoretic properties of $\mu$ are used.
\begin{lemma}[Equation (13) and (14) in Section 2.5 of \cite{schmidt1994subgroup}]\label{lem:schmidt-mu}
    For all $i,n \in \mathbb{N}$, $p^i\mid n$ if and only if $p^i\mid \mu(n)$. Moreover, for all $x,y,m\in\mathbb{N}$, $x=y\pmod{p^m}$ if and only if $\mu(x)=\mu(y)\pmod{p^m}$.
\end{lemma}
Since $\mu$ is a permutation of the elements in $\mathbb{Z}_q$, the inverse $\nu=\mu^{-1}:\mathbb{Z}_q\to \mathbb{Z}_q$ exists. In our algorithm, we will require that $\nu$ can be evaluated efficiently. The next lemma shows that this is the case.
\begin{lemma}[Efficient inversion of $\mu$]
\label{lem:efficient-mu-inverse}
Let $q=p^m$,$r=1+p^s$ and $\mu(j)=\sum_{u=0}^{j-1}r^u$
where $s\geq 1$, and $s\geq 2$ when $p=2$.
Given
$i\in\{0,\ldots,q-1\}$, the unique
$j\in\{0,\ldots,q-1\}$ satisfying
$$
\mu(j)\equiv i\pmod q
$$
can be computed in time polynomial in $\log q$.
\end{lemma}
\begin{proof}
We first show that, for every $h\geq 1$, every integer $x$, and
every $t\in\{0,\ldots,p-1\}$,
\begin{align}\label{eq:mu-sum}
    \mu(x+tp^h)
\equiv
\mu(x)+tp^h
\pmod{p^{h+1}}.
\end{align}
Using $r=1+p^s$, we have
\[
\mu(n)
=
\frac{(1+p^s)^n-1}{p^s}
=
n+\sum_{a=2}^{n}\binom{n}{a}p^{s(a-1)}.
\]
Where in the first equality we have just summed up the terms of $\mu$ and in the second equality we have used the binomial theorem.
For a nonzero integer $M$, let $v_p(M)$ be the largest nonnegative integer $e$ such that $p^e \mid M$. Intuitively, $v_p(M)$ counts how many factors of $p$ occur in $M$.
If $n=tp^h$, then, for $a\geq 2$,
\begin{align}\label{eq:vp_juggle}
v_p\left(\binom{n}{a}p^{s(a-1)}\right)&=v_p\left(\binom{n}{a}\right) + v_p\left( p^{s(a-1)}\right),\nonumber \\ \nonumber
&= v_p\left(\binom{n}{a}\right) +s(a-1),\\ \nonumber
&\geq h-v_p(a)+s(a-1),\\
&\geq h+1.
\end{align}
The third line is obtained by the fact that $v_p\left(\binom{n}{a}\right)\geq h-v_p(a)$ which follows from two facts. First,  $v_p(n)\geq h$. Second, consider the binomial identity
\begin{align}
    a\binom{n}{a}=n\binom{n-1}{a-1}.
\end{align}
Which implies that
\begin{align}
    v_p\left(\binom{n}{a}\right)=v_p(n)-v_p(a)+v_p\left(\binom{n-1}{a-1} \right),
\end{align}
since the third term is $\geq 0$ and $v_p(n)\geq h$ we get the inequality of the third line in \cref{eq:vp_juggle}. The fourth line follows from the fact that $s(a-1)-v_p(a)\geq 1$. 

Consequently, with \cref{eq:vp_juggle}, we have shown that $p^{h+1}\mid \binom{n}{a}p^{s(a-1)}$ and then $\mu(n)-n$ is divisible by $p^{h+1}$. Thus,
\[
\mu(tp^h)\equiv tp^h\pmod{p^{h+1}}.
\]
Moreover, $\mu(x+y)=\mu(x)+r^x\mu(y)$ from the definition of $\mu$. Since $r^x\equiv1\pmod p$, taking $y=tp^h$ proves \cref{eq:mu-sum}.

We now show that $j$ can be recovered one base-$p$ digit at a time. Set
\[
j_1:=i\bmod p.
\]
Since $\mu(x)\equiv x\pmod p$, we have
\[
\mu(j_1)\equiv i\pmod p.
\]

Suppose inductively that $j_h$ has been computed and satisfies
\[
\mu(j_h)\equiv i\pmod{p^h}.
\]
Define
\[
\kappa_h
:=
\frac{i-\mu(j_h)}{p^h}
\pmod p
\]
and set
\[
j_{h+1}:=j_h+\kappa_hp^h.
\]
The quotient defining $\kappa_h$ is well defined modulo $p$ by the
induction hypothesis. Using \cref{eq:mu-sum},
\[
\mu(j_{h+1})
\equiv
\mu(j_h)+\kappa_hp^h
\equiv
i
\pmod{p^{h+1}}.
\]
After $m-1$ lifting steps, we obtain $j_m$ satisfying
\[
\mu(j_m)\equiv i\pmod{p^m}.
\]
By Lemma~\ref{lem:schmidt-mu}, $\mu$ is a permutation modulo
$p^m$, so this $j_m$ is the unique value $\nu(i)$.

At the $h$-th step, only the residue
$\mu(j_h)\bmod p^{h+1}$ is required. It can be computed by repeated
squaring using
\[
\begin{pmatrix}
r & 1\\
0 & 1
\end{pmatrix}^{j_h}
=
\begin{pmatrix}
r^{j_h} & \mu(j_h)\\
0 & 1
\end{pmatrix}
\pmod{p^{h+1}}.
\]
There are $m=O(\log q)$ lifting steps, each involving arithmetic on
$O(\log q)$-bit integers. Hence the computation takes time
polynomial in $\log q$.

\end{proof}
The following property is useful in defining an automorphism in $A$.
\begin{lemma}
$\tau:=\mu(q)/q$ is not divisible by $p$.
\end{lemma}
\begin{proof}
    First, we note that $\tau$ is an integer. Since $p^m\mid q$ then $p^m\mid\mu(q)$ by \cref{lem:schmidt-mu}. Note also that $p^{m+1}$ does not divide $q$, therefore $p^{m+1}$ does not divide $\mu(q)$.
\end{proof}
Therefore, the function on $A$ that exponentiates by $\tau$ is an automorphism of $A$.
On the other hand, since $bA\in P/A$ has order $q$, then $b^q\in A$. Given $b$, choose the unique $a_0\in A$ such that $a_0^\tau = b^q$.

Now, we define the Abelian $p$-group $B$. Consider an extra element $c$. More specifically, $c$ is a formal symbol with which we define $B$ in the following way
\begin{align}\label{eq:B-group}
    B=\langle A,c \vert ac=ca \text{ for all } a\in A, c^q=a_0   \rangle
\end{align}
In other words, we consider the group $A$ and add the symbol $c$ such that it commutes with all $a$ and $c^q=a_0$. Note that this group is well defined, any element of this group can be written as $ac^i$ with $a\in A$ and $0\leq i < q$ and the product of these elements is determined by the defining
relations.

Next, define the map $\sigma:B\to P$ by $\sigma(ac^{\mu(j)})=b^ja$ for $a\in A$. We remark that this map is not a group homomorphism, but a crossed isomorphism. 
\begin{definition}
    Let $G$ and $G'$ be groups,  a map $\sigma:G\to G'$ is a crossed isomorphism if $\sigma$ is a bijective map such that $\sigma(x)\sigma(y)=\sigma(f(y)(x)y)$ where $f:G\to\mathrm{End}(G)$.
\end{definition}
 Intuitively, a crossed isomorphism can be thought of as a ``twisted isomorphism'', where the map $f$ defines the twist. When $f$ maps every element to the identity, then $\sigma$ is an isomorphism. The map $\sigma:B\to P$ defined in the previous paragraph corresponds to a crossed isomorphism for a function $f$ which we define later. For the purpose of solving the HSP, we use the property of $\sigma$ being a crossed isomorphism to prove that for all $K\leq B$ and $x\in B$, $\sigma(Kx)=\sigma(K)\sigma(x)$. 
 
Consider $\alpha:B\to B$ such that $\alpha(x)=x^r$. Note that $B$ is an Abelian $p$-group and thus exponentiation by $r$ is an automorphism. Consider an element $x=ac^i\in B$ for some $0\leq i < q$, note that there exists $0\leq j<q$ such that $\mu(j)= i \pmod{q}$. Then, define $f:B\to \mathrm{Aut}(B)$ as $f(x)=\alpha^{j}$. The map $f$ is well defined because if $j,j'$ satisfy $\mu(j)=\mu(j') \pmod{q}$, then $j=j' \pmod{q}$ and thus $\alpha^q$ is the identity. This shows that the definition of $f$ is independent of the choice $j$ or $j'$.

\begin{lemma}
    The map $\sigma:B\to P$ defined by $\sigma(ac^{\mu(j)})=b^{j}a$ is a crossed isomorphism.
\end{lemma}
\begin{proof}
    Define the group $B^*=(B,\ast)$ which consists of the elements in $B$ and the product $\ast$ defined by $x\ast y=f(y)(x)\cdot y$, where $\cdot$ denotes the original product in $B$. Let $\iota: B \to B^*$ be the identity function on the underlying sets of $B$ and $B^*$ and $\theta:B^*\to P$ be the isomorphism sending $A$ identically to $A$ and $c$ to $b$. 
    
    We claim that $\sigma=\theta\circ\iota$, which we prove next. Note that in the group $(B,\ast)$, $c^{\ast j}=c^{\mu(j)}$.  For $j=1$ we have that $c=c^{\mu(1)}$ since $\mu(1)=1$. Assume $c^{\ast j}=c^{\mu(j)}$, then
    \begin{align}
     c^{\ast(j+1)}&=c^{\ast j}\ast c,\\ 
     &=f(c)(c^{\mu(j)})\cdot c,\\
     &= \alpha(c^{\mu(j)})\cdot c,\\
     &=(c^{\mu(j)})^{r}\cdot c,\\
     &=c^{r\mu(j)+1}.
    \end{align}
    Where we used the fact that $f(c)=\alpha$. But note that $\mu(j+1)=1+r+\cdots+r^{j}=1+r(1+r+\cdots + r^{j-1})=1+r\mu(j)$. Thus, $c^{\ast(j+1)}=c^{\mu(j+1)}$. Take $a\in A$, then $f(a)$ corresponds to the identity map. Therefore,
    \begin{align}
        c^{\ast j}\ast a = f(a)(c^{\mu(j)})a = c^{\mu(j)}a=ac^{\mu(j)}.
    \end{align}
    Since $\theta$ is a group isomorphism, $\theta(c^{\ast j}\ast a)=\theta(c)^j\theta(a)=b^j a$. This shows that $\sigma(ac^{\mu(j)})=b^j a$ as claimed.

    We now prove that $\sigma$ is a crossed isomorphism. Consider $\sigma(f(y)(x)\cdot y)=\theta(\iota(f(y)(x)\cdot y))$. We have that $\iota(f(y)(x)\cdot y)=\iota(x)\ast\iota(y)$, and thus 
    \begin{align}
     \sigma(f(y)(x)\cdot y)&=\theta(\iota(x)\ast\iota(y)),\\
     &=\theta(\iota(x))\theta(\iota(y)),\\
     &=\sigma(x)\sigma(y).
    \end{align}
    where the second line follows from the fact that $\theta$ is an isomorphism. 
\end{proof}
To prove that $\sigma(Kx)=\sigma(K)\sigma(x)$, we use the following theorem.
\begin{theorem}[Baer's theorem \cite{Baer1944CrossedIsomorphisms}. Theorem 2.5.8 in \cite{schmidt1994subgroup}]\label{thm:schmidt:baers-258}
    Let $G$ and $G'$ be finite groups, $f:G\to \mathrm{End}(G)$ and suppose that $\sigma:G\to G'$ is a crossed isomorphism with respect to $f$. Then the following are equivalent:
    \begin{enumerate}
        \item The map $\sigma$ induces a projectivity from $G$ to $G'$ and satisfies $\sigma(Kx)=\sigma(K)\sigma(x)$ for all $K\leq G$ and $x\in G$.
        \item For all $x,y\in G$, the following conditions are satisfied.
        \begin{enumerate}
            \item $f(x)$ is a power automorphism, i.e., $f(x)\in \{g\in \mathrm{Aut}(G): g(K)=K \text{ for all } K\leq G\}$.
            \item If $o(x)=2^m$ with $m>1$, then $f(x)(x)=x^{1+4i}$ with $i\in \mathbb{N}$.
            \item If $xy=yx$ and $o(x)$ and $o(y)$ are finite and relatively prime, then $f(y)(x)=x$.
        \end{enumerate}
    \end{enumerate}
\end{theorem}
Therefore, we need to check that condition $(2)$ above is true for $f$ to prove that $\sigma$ satisfies $\sigma(Kx)=\sigma(K)\sigma(x)$.
\begin{lemma}
    The map $f:B\to\mathrm{Aut}(B)$ fulfills condition (2) of \cref{thm:schmidt:baers-258}.
\end{lemma}
\begin{proof}
    Recall that for $x=ac^{\mu(j)}\in B$, we have that $f(x)=\alpha^j$ with $\alpha(y)=y^r$ for $y\in B$. As remarked before, $\alpha$ is an automorphism of $B$ and now we show that $\alpha(K)=K$ for $K\leq B$. Clearly $\alpha(K)\subseteq K$. The inverse of $\alpha$ is given by $\alpha^{-1}(x)=x^{r^{-1}}$ where $rr^{-1}=1 \pmod{p^n}$. Note that $\mathrm{Exp}(B)=p^n$ for the following reason.
    The element \(b^q\) has order \(p^s\).  Since
    $$
    a_0^\tau=b^q
    $$
    and \(p\nmid\tau\), the element \(a_0\) also has order \(p^s\).
    The relation \(c^q=a_0\), together with $q=p^m$,
    therefore implies that $c$ has order
    $p^{m+s}=p^n$.
    Since \(\operatorname{Exp}(A)\leq p^n\), it follows that $\operatorname{Exp}(B)=p^n$.
    

    This implies that $\alpha^{-1}(K)\subseteq K$ and thus $K \subseteq \alpha(K)$ proving that $f(x)$ is a power automorphism.

    For property (b), note that it is only relevant to check when $p=2$ because otherwise $p$ is odd, and therefore there are no elements in $b$ with a power of $2$ order. Given $p=2$, by \cref{prop:p-group-modular}, it is true that $s\geq 2$ and therefore $2^s = 0 \pmod{4}$ and then $1+2^s=1 \pmod{4}$. Recall that we defined $r=1+2^s$ and also note $r^j=1 \pmod{4}$ for any $j\geq 0$. Thus, there is an integer $i$ such that $r^j=1+4i$. Finally, $f(x)(x)=x^{r^j}=x^{1+4i}$.

    For property (c), we have that if two elements in a $p$-group have relatively prime orders, then one of them is trivial and the condition is fulfilled.
\end{proof}
It follows that $\sigma(Kx)=\sigma(K)\sigma(x)$ for all $K\leq B$ and $x\in B$. Recall that the HSP oracle $F:P\to S$ in terms of left cosets, so we will define $\lambda:B\to P$, $\lambda(x)=\sigma(x^{-1})^{-1}$ (we could as well consider a right coset HSP oracle). Note that $\lambda(K)=\sigma(K)$, and hence the two maps induce the same
projectivity on subgroups and moreover $\lambda(xK)=\lambda(x)\lambda(K)$. Since $F$ hides the subgroup $H$, define $\widetilde{H}=\lambda^{-1}(H)$ and also the oracle $\widetilde{F}(x)=F(\lambda(x))$. The oracle $\widetilde{F}$ hides the group $\widetilde{H}\leq B$. 

Finally, we show that it is enough to run the Abelian HSP algorithm on group $B$ with oracle $\widetilde{F}$ to recover the generators of $H$. We begin by describing how to implement the oracle $\widetilde{F}$ given access to $F$.

\subsection{Construction of group $B$ and efficient implementation of algorithm for HSP}\label{sec:construct-B}

We are given as input the group $P=A\langle b \rangle$ with $A\unlhd P$, $A$ Abelian and $P/A=\langle bA \rangle$, $\abs{P:A}=q=p^m$. The action of $b$ on $A$ is given by $b^{-1}ab=a^r$ with $r=1+p^s$.

Concretely, we are given the description of the generators of $A$ and $b$ according to some encoding. The encoding of $A$ can be obtained by considering the decomposition $A\cong \mathbb{Z}_{p^{e_1}}\times \cdots \mathbb{Z}_{p^{e_d}}$, an element of $A$ is encoded as a vector $(\alpha_1,\cdots,\alpha_d)$ with $\alpha_\ell\in \mathbb{Z}_{p^{e_\ell}}$. Recall that $b^{q}\in A$ and therefore an element $b^j a$ of $P$ can be encoded as $\ket{j,a}$ with $0\leq j < q$ and $a\in A$. Note that arithmetic in this group can be easily handled using the identity $b^{-1}ab=a^r$.

Recall that $a_0$ is defined by the identity $a_0^{\tau}=b^{q}$ with $\tau=\mu(q)/q$. Since exponentiation by $\tau$ defines an automorphism, it has an inverse. We can compute such inverse and obtain $a_0$. We then proceed to encode the group $B$. Every element of this group can be written as $ac^i$ for $0\leq i<q$. Thus, we encode $B$ with the basis states $\ket{a,i}$. Arithmetic can be done efficiently in $B$. Consider $a_1 c^{i_1}$ and $a_2 c^{i_2}$, then the multiplication is given by $a_1 a_2 c^{i_1 + i_2}$. Write $i_1 + i_2= \ell q + k$ with $0\leq k<q$. Then $c^{i_1+i_2}=c^{ql}c^k=a_0^\ell c^k$. Therefore $\ket{(a_1,i_1)(a_2,i_2)}=\ket{a_1 a_2 a_0^{\ell},k}$, showing that the arithmetic can be done efficiently. Note that inverses can also be obtained efficiently.
\begin{align}\label{eq:B-inverse}
    (a,i)^{-1}=
    \begin{cases}
    (a^{-1},0) & \text{if } i=0, \\
    (a^{-1}a_0^{-1},q-i)  & \text{if } 0< i <q. 
\end{cases}
\end{align}
To solve the HSP in $B$, we need to implement the quantum Fourier transform. The usual implementation requires $B$ to be decomposed into a direct product of cycle groups. Such decomposition can be found using the Smith normal decomposition as in \cite{cheung2001decomposing} to decompose any Abelian group. More concretely, we are describing the elements of $B$ through the basis given by $\ket{a,i}$ with $a\in A$ and $i\in \mathbb{Z}_q$ and we can write $\ket{a,i}=\ket{\alpha_1,\cdots, \alpha_d, i}$ with $\alpha_\ell \in \mathbb{Z}_{p^{e_\ell}}$, we call this basis the ``canonical basis''. To implement the quantum Fourier transform, we need the states to be in the form $\ket{y}:=\ket{y_1,\cdots, y_t}$ with $y_\ell \in \mathbb{Z}_{N_\ell}$ and where $B\cong \mathbb{Z}_{N_1}\times \cdots \times \mathbb{Z}_{N_t}$. We refer to this basis as the ``Smith normal form (SNF) basis''.


Given $B$ in terms of generators $g_1,\cdots g_d, c$ according to \cref{eq:B-group}, to obtain the cyclic decomposition of $B$, we encode the defining
relations of $B$ in an integer relation matrix and compute its Smith
normal form, following the algorithm of Theorem~7 in~\cite{cheung2001decomposing}.

We can efficiently obtain $a_0$ in terms of the generators of $A$, i.e., we get $a_0= g_1^{u_1}\cdots g_d^{u_d}$ with $0\leq u_i < p^{e_i} $. Then, to define $B$, we have the relations $g_i^{p^{e_i}}=1$ and $c^{q}=g_1^{u_1}\cdots g_d^{u_d}$. For simplicity, we define $g_{d+1}:=c$. Let $\eta=d+1$, the vector $x=(\alpha_1,\cdots, \alpha_d, i)^T\in \mathbb{Z}^\eta$ represents an element of $B$. We can then define the following matrix in terms of the relations described above
\begin{align}
    M=\begin{bmatrix}
    p^{e_1} & 0 & \cdots & 0 & -u_1 \\
    0 & p^{e_2} & \cdots & 0 & -u_2 \\
    \vdots & \vdots & \ddots & \vdots & \vdots \\
    0 & 0 & \cdots & p^{e_d} & -u_d \\
    0 & 0 & \cdots & 0 & q
\end{bmatrix}.
\end{align}
The next theorem from \cite{cheung2001decomposing} shows that we can obtain the generators $\tilde{g}_1,\cdots, \tilde{g}_t$ from the matrix $M$.
\begin{theorem}[Theorem 7 in \cite{cheung2001decomposing}]
    Given a generating set $\{g_1,\cdots, g_d\}$ of a finite Abelian group $G$ and a matrix $M$ such that $g_1^{x_1}\cdots g_d^{x_k}=e$ if and only if $x=(x_1,\cdots,x_k)^T \in \mathrm{intcol}(M)$ where $\mathrm(M)$ denotes the set of vectors obtainable by taking integer linear combinations of columns of $M$, we can find in polynomial time (in the size of $M$) $\tilde{g}_1,\cdots, \tilde{g}_t$ with $t\leq k$ such that $G=\langle \tilde{g}_1\rangle \oplus \cdots \oplus \langle \tilde{g}_t \rangle $. 
\end{theorem}
The proof of this theorem proceeds by using the SNF to derive the generators $\tilde{g}_i$. Given the matrix $M$, one can find in polynomial time unimodular matrices (i.e., nonsingular integer matrices of determinant $\pm$ 1) $U$ and $V$, such that $U^{-1}MV=D=\mathrm{diag}(N_1,N_2,\cdots,N_t)$. 
Note that $M$ is such that $g_1^{x_1}\cdots g_d^{x_d} c^{x_{d+1}}=e$ if and only if $x=(x_1,\cdots,x_{d+1})\in \mathrm{intcol}(M)$. Moreover, $\mathrm{intcol}(M)=\mathrm{intcol}(MV)$. 
Define $\tilde{g_i}:=g_1^{U_{1,i}}\cdots g_{d+1}^{U_{k,i}}$, and then $\tilde{g}_1^{x_1}\cdots \tilde{g}_{d+1}^{x_d}=e$ if and only if $Ux\in \mathrm{intcol}(M)$. The last statement is true if and only if $x\in\mathrm{intcol}(U^{-1}MV)$. Since $U^{-1}MV=D$ we have that the $\tilde{g}_i$ are the generators of the subgroups in the decomposition of $B$ into cyclic groups. Note that $U$ also gives the transformation between the canonical basis and the SNF basis, i.e., if we let $y$ be the coordinates in the SNF basis and $x$ in the canonical basis, then $y_i=(U^{-1}x)_i \pmod{N_i}$.
This means that if we have the quantum state $\ket{x}=\ket{\alpha_1,\cdots,\alpha_d,c}$, we can transform it into the SNF basis by acting with a unitary $W$ such that
\begin{align}
 W\ket{x}=\ket{(U^{-1}x)_1 \pmod{N_1},\cdots (U^{-1}x)_{t}\pmod{N_t}}.   
\end{align}
 The unitary $W$ can be implemented efficiently since $U^{-1}$ can be implemented with a quantum circuit implementing the arithmetic. 
\subsection{Implementing $\widetilde{F}$}
The problem of implementing $\widetilde{F}$ reduces to that of implementing $\lambda:B\to P$. Recall $\lambda(x)=\sigma(x^{-1})^{-1}$. As already mentioned in \cref{sec:construct-B}, $x^{-1}$ can be computed efficiently. The crossed isomorphism $\sigma$ is computed according to $\sigma(ac^{\mu(j)})=b^j a$ for $a\in A$ and $0\leq j <q$.
Given $x^{-1}=ac^i$ with $a\in A$ and $0\leq i<q$, we first compute $j=\nu(i)$ using Lemma~\ref{lem:efficient-mu-inverse}. Now, we need to compute $\mu(j)$ to give as input to $\sigma$ the element $a c^{\mu(j)}$.

Let $E_A:=\operatorname{Exp}(A)$.
Since $A$ is given in invariant-factor coordinates, $E_A$ is known
from the input. We compute instead
\[
\widehat{\mu}(j):=\mu(j)\bmod qE_A,
\qquad
0\leq\widehat{\mu}(j)<qE_A.
\]
This is because the exponent of $c$ is only relevant modulo $qE_A$, since every $q$ powers of $c$ produce one power of $a_0$ and these powers repeat modulo $E_A$.
This residue can be obtained by repeated squaring using
\[
\begin{pmatrix}
r & 1\\
0 & 1
\end{pmatrix}^{j}
=
\begin{pmatrix}
r^j & \mu(j)\\
0 & 1
\end{pmatrix}
\pmod{qE_A}.
\]

Since $\mu(j)\equiv i\pmod q$, there is a unique
$\ell\in\{0,\ldots,E_A-1\}$ such that
$$
\widehat{\mu}(j)=i+\ell q.
$$
Equivalently,
\[
\ell
\equiv
\frac{\mu(j)-i}{q}
\pmod{E_A}.
\]

There is therefore an integer $z$ such that $\mu(j)=i+q(\ell+zE_A)$.
Using $c^q=a_0$ and $a_0^{E_A}=e$, we obtain $
c^{\mu(j)}
=
c^i(c^q)^{\ell+zE_A}
=
c^i a_0^{\ell}$.

Consequently,
\[
x^{-1}
=
ac^i
=
aa_0^{-\ell}c^{\mu(j)}.
\]
Applying the crossed isomorphism gives
\[
\sigma(x^{-1})
=
\sigma\left(aa_0^{-\ell}c^{\mu(j)}\right)
=
b^j aa_0^{-\ell}.
\]
Finally, $\lambda(x)=\sigma(x^{-1})^{-1}$.
All the operations above consist of modular addition,
multiplication, division by the known power $q$, and repeated
squaring on integers with $O(\log q+\log E_A)$
bits. Hence they can be implemented by a reversible quantum circuit
of size polynomial in $\log q+\log E_A$. Ancillary registers holding
$j$, $\widehat{\mu}(j)$, and $\ell$ can be uncomputed after
$\lambda(x)$ has been evaluated.
Assuming oracle access to $F$, we can get an oracle for $\widetilde{F}$ by using the circuit described above for $\lambda$ and then applying $F$.

\subsection{Recovering the hidden subgroup of $P$}
Let $H\leq P$ be hidden subgroup of $F$ and define $\widetilde{H}:=\lambda^{-1}(H)\leq B$. Note that $\lambda^{-1}(H)$ is a subgroup because $\lambda$, via $\sigma$, induces a projectivity between subgroup lattices. We claim that $\widetilde{F}$ hides $\widetilde{H}$. 
\begin{lemma}
    For $x,y\in B$, $\widetilde{F}(x)=\widetilde{F}(y)$ if and only if $x\widetilde{H}=y\widetilde{H}$.
\end{lemma}
\begin{proof}
    Let $x,y\in B$, $\widetilde{F}(x)=\widetilde{F}(y)$. This is true if and only if $F(\lambda(x))=F(\lambda(y))$. Since $F$ hides $H$, then this holds iff 
    \begin{align}
        \lambda(x)H=\lambda(y)H.
    \end{align}
    Since $H=\lambda(\widetilde{H})$, and $\lambda$ preserves left cosets,
    \begin{align}
    \lambda(x\widetilde{H})=\lambda(x)\lambda(\widetilde{H})=\lambda(x)H.    
    \end{align}
    Thus, $\lambda(x\widetilde{H})=\lambda(y\widetilde{H})$. The map $\lambda$ is injective, therefore $x\widetilde{H}=y\widetilde{H}$.
\end{proof}
 By running standard algorithm for the Abelian HSP in $B$, we can recover the generators of $\widetilde{H}$. The algorithm returns the generators $k_1,\cdots,k_{\ell}$ for $\widetilde{H}$. For each generator compute $h_i=\lambda(k_i)$. The $\{h_i\}$ generate $H\leq P$. The reason is that $\sigma$, and therefore $\lambda$, induces a projectivity on the subgroup lattice by \cref{thm:schmidt:baers-258}. Thus, 
 \begin{align}
  H=\lambda(\widetilde{H})=\lambda\left( \bigvee_i \langle k_i \rangle \right) = \bigvee_i \lambda(\langle k_i \rangle ).
 \end{align}
Moreover, $\lambda(\langle k_i \rangle)=\langle \lambda(k_i) \rangle$. Since $\lambda(\langle k_i \rangle)$ is a subgroup containing $\lambda(k_i)$, it follows that $\langle \lambda(k_i) \rangle \leq \lambda(\langle k_i \rangle)$. On the other hand, $\lambda^{-1}(\langle \lambda(k_i) \rangle)$ contains $k_i$. Therefore,
$$\langle k_i \rangle \leq \lambda^{-1}(\langle \lambda(k_i) \rangle).$$
Applying $\lambda$ on both sides yields $\lambda(\langle k_i \rangle)\leq \langle \lambda(k_i) \rangle$ and thus $H=\langle \lambda(k_1), \cdots , \lambda(k_{\ell})\rangle$.
\section*{Acknowledgements}
We thank G\'abor Ivanyos for his insightful comments on an early version of this manuscript. We also thank François Le Gall for pointing out relevant papers to this work and Andrew Childs for helpful comments on this manuscript. The author acknowledges support from the U.S. Department of Defense through a QuICS Hartree Fellowship.
\printbibliography
\newpage
\appendix
\section{Comparison with previous results}\label[appendix]{sec:comparison}

Before comparing the different results, it is useful to distinguish
two possible meanings of one result extending another.  First, one
may compare the classes of groups for which the algorithms
are applicable. Second, one may compare the
input models: an algorithm that works for an arbitrary black-box
presentation is stronger in this respect than one which assumes that a
particular decomposition of the group is supplied.

\subsection{Comparison with earlier semidirect-product algorithms}
\label{app:comparison-semidirect}

\paragraph{The result of Inui and Le Gall.}
Inui and Le Gall~\cite{Inui2007semimdirect} have given  groups a polynomial-time algorithm for the groups of the form
$$
    P_{p,r}
    :=
    \left\langle x,y
    \,\middle|\,
    x^{p^r}=y^p=1,\quad
    yx=x^{1+p^{r-1}}y
    \right\rangle .
$$
with $p$ prime, $r\geq 2$ and when $p=2$ then $r\neq 2$.
This group corresponds to $\mathbb{Z}_{p^r}\rtimes\mathbb{Z}_{p}$ such that $yxy^{-1}=x^{1+p^{r-1}}$.

Note that this family of groups is in fact a special case of case (2) in \cref{prop:p-group-modular}. Therefore, our result applies to a broader family of groups. On the other hand, the method by Inui and Le Gall does not require a unique encoding. In case (2) of \cref{prop:p-group-modular}, we need a unique encoding of $A$ to use the algorithm of \cite{cheung2001decomposing} and decompose $A$.

Inui and Le Gall also consider $\left(\mathbb Z_{p^r}\right)^m\rtimes \mathbb Z_p$,
where the generator of \(\mathbb Z_p\) acts on every coordinate by
multiplication by \(1+p^{r-1}\).  In this case they construct a
bijection
\[
    \pi:
    \left(\mathbb Z_{p^r}\right)^m\rtimes\mathbb Z_p
    \longrightarrow
    \left(\mathbb Z_{p^r}\right)^m\times\mathbb Z_p
\]
which is not a group isomorphism, but nevertheless maps subgroups to
subgroups and cosets to cosets.  This reduces the problem to an Abelian
HSP. The  result assumes unique encoding and that
generators for the normal factor and the cyclic complement are supplied
separately.

These groups are special cases of the quasi-Hamiltonian groups
considered in \cref{thm:main-quasihamiltonian}.  Indeed, let
$A=\langle x\rangle\cong\mathbb Z_{p^r}$
and put $b=y^{-1}$.  Then
$$
    b^{-1}xb=yxy^{-1}=x^{1+p^{r-1}}.
$$
Thus $P_{p,r}=A\langle b\rangle$, the quotient $P_{p,r}/A$
is cyclic, and $b$ acts on $A$ by the power automorphism $a\longmapsto a^{1+p^s}$ with $s=r-1$.
This is falls under case (2) of \cref{prop:p-group-modular}.
The same observation applies to
\((\mathbb Z_{p^r})^m\rtimes\mathbb Z_p\), because the same power
automorphism acts on every coordinate.

Thus, the groups
treated by Inui and Le Gall form a proper subclass of the finite
quasi-Hamiltonian groups.  The containment is strict.  For example,
\[
    \mathbb Z_{81}\rtimes_{4}\mathbb Z_{27},
\]
where the generator of \(\mathbb Z_{27}\) acts by multiplication by
\(4=1+3\) modulo \(81\), is quasi-Hamiltonian, but it is not in the class considered by Inui and Le Gall.

\subsection{Comparison with poly-near-Hamiltonian groups}
\label{app:comparison-near-hamiltonian}

The terminology ``poly-near-Hamiltonian'' used by
Gavinsky~\cite{Gavinsky2004nearHam} refers to a different group-theoretic
property from the property of being quasi-Hamiltonian.  For a subgroup \(K\leq G\), let
\[
    N_G(K)
    :=
    \{g\in G:gKg^{-1}=K\}
\]
be its normalizer, and define
\[
    M_G
    :=
    \bigcap_{K\leq G}N_G(K).
\]
Thus \(M_G\) consists of the elements which normalize every subgroup
of \(G\).  If \(G\) is Dedekind, then every subgroup is normal and
hence \(M_G=G\).

Writing $n:=\log |G|$, Gavinsky calls a family poly-near-Hamiltonian when
$$
    [G:M_G]=\operatorname{poly}(n).
$$
This extends the earlier result by
Grigni, Schulman, Vazirani and Vazirani~\cite{Grigni2001HSP}, whose
corresponding condition was
\[
    [G:M_G]
    \in
    \exp\!\left(O\!\left(\sqrt{\log n}\right)\right).
\]




\paragraph{Incomparability with quasi-Hamiltonian groups.}
The poly-near-Hamiltonian and quasi-Hamiltonian conditions are
incomparable.
Since poly-near-Hamiltonianity is an asymptotic
condition, the comparison below is stated for growing families of
groups.

First, fix an odd prime \(p\), and for \(k\geq 2\) define
\[
    G_k
    :=
    \left\langle a,b
    \,\middle|\,
    a^{p^k}=1,\quad
    b^{p^{k-1}}=1,\quad
    b^{-1}ab=a^{1+p}
    \right\rangle .
\]
Equivalently,
\[
    G_k\cong
    \mathbb Z_{p^k}
       \rtimes_{1+p}
    \mathbb Z_{p^{k-1}}.
\]
The action is a power automorphism of the form appearing in \cref{prop:p-group-modular}.  Hence \(G_k\) is a modular
\(p\)-group and is therefore quasi-Hamiltonian.

Let $C:=\langle b\rangle$. We will show that $M_G$ has exponentially large index by proving that $C$ has a small normalizer.  Since powers of $b$ normalize
$C_k$, it is enough to determine which powers of $a$ normalize
$C_k$.  Using the defining relation of $G_k$, we have that
\[
    a^{-u}ba^u
    =
    a^{u((1+p)^{-1}-1)}b,
\]
where the inverse is taken modulo \(p^k\).  Note that the integer
\((1+p)^{-1}-1\) is divisible by $p$, but not by \(p^2\).
It is divisible by $p$ because $(1+p)(1+p)^{-1}=1 \pmod{p^k}$ and therefore
\begin{align}\label{eq:modpk}
    (1+p)(1+p)^{-1}=1+m p^k.
\end{align}
From this equality we get $(1+p)^{-1}-1=mp^k -p(1+p)^{-1}$, which implies that $(1+p)^{-1}-1$ is divisible by $p$. To see that it is not divisible by $p^2$, assume that $(1+p)^{-1}=1\pmod{p^2}$. Then $(1+p)(1+p)^{-1}=1+p \pmod{p^2}$. But by the inverse property $(1+p)(1+p)^{-1}=1\pmod{p^2}$. This implies $(1+p)=1\pmod{p^2}$ and therefore $p=0 \pmod{p^2}$ which is impossible.

Consequently,
\[
    a^u\in N_{G_k}(C_k)\quad\Longleftrightarrow\quad p^k\mid u((1+p)^{-1}-1)
    \quad\Longleftrightarrow\quad
    p^{k-1}\mid u.
\]
It follows that $N_{G_k}(C_k)=\left\langle a^{p^{k-1}},b\right\rangle$ and therefore
\[
    [G_k:N_{G_k}(C_k)]
    =
    p^{k-1}.
\]
Since
\[
    M_{G_k}\leq N_{G_k}(C_k),
\]
we obtain
\[
    [G_k:M_{G_k}]
    \geq p^{k-1}.
\]
For fixed \(p\),
\[
    \log |G_k|
    =
    (2k-1)\log p
    =
    \Theta(k),
\]
whereas \(p^{k-1}\) is exponential in \(k\).  Thus
\(\{G_k\}\) is a quasi-Hamiltonian family which is not
poly-near-Hamiltonian.

For the converse direction, let
\[
    D_4
    =
    \left\langle r,s
    \,\middle|\,
    r^4=s^2=1,\quad srs=r^{-1}
    \right\rangle
\]
be the dihedral group of order \(8\), according to our convention,
and define
\[
    F_m:=D_4\times\mathbb Z_2^m.
\]
The center of \(F_m\) has index \(4\):
\[
    Z(F_m)
    =
    \langle r^2\rangle\times\mathbb Z_2^m,
    \qquad
    [F_m:Z(F_m)]=4.
\]
Since every central element normalizes every subgroup,
\[
    Z(F_m)\leq M_{F_m},
\]
and hence
\[
    [F_m:M_{F_m}]\leq 4.
\]
Thus \(\{F_m\}\) is a poly-near-Hamiltonian family.

It is not quasi-Hamiltonian.  Inside the \(D_4\) factor, consider
\[
    H:=\langle s\rangle,
    \qquad
    K:=\langle rs\rangle.
\]
Then
\[
    HK=\{1,s,rs,r^3\},
    \qquad
    KH=\{1,s,rs,r\},
\]
so \(HK\neq KH\).  These subgroups remain nonpermutable after being
embedded in \(F_m\).

\subsection{Comparison with bounded nilpotency class}\label{app:nilpotent-comparison}
In \cite{Ivanyos2012Nil2} it is shown that there is an efficient quantum algorithm for groups of nilpotency class 2. These groups have derived subgroup contained in the center, i.e., if $G$ is a group of nilpotency class 2 then $[G,G]\leq Z(G)$. In a later work \cite{Imran2024}, it has been shown that the HSP can be solved efficiently for groups with nilpotency class and prime factors of the group order bounded by constants. We show below that the groups studied in this paper have unbounded nilpotency class.

More generally, the nilpotency class is defined by the lower central series. Let $\gamma_i(G)$ be defined recursively by $\gamma_1(G)=G$ and $\gamma_{i+1}(G)=[\gamma_i(G),G]$. This gives a descending chain of normal subgroups,
\begin{align}
    G=\gamma_1(G)\unrhd \gamma_2(G)\unrhd \cdots
\end{align}
A group $G$ is nilpotent if its lower central series eventually becomes trivial, i.e., there is a $c\geq 0$ such that $\gamma_{c+1}(G)=\{e\}$. In this case, we say that the group has class $c$. 

We first compare the quasi-Hamiltonian groups with the groups of
bounded nilpotency class considered in previous work.  Fix a prime
\(p\) and integers
\[
1\leq s<k,
\]
where \(s\geq2\) when \(p=2\), and define
\[
P_{k,s}^{(p)}
:=
\left\langle
a,b
\;\middle|\;
a^{p^k}=1,\quad
b^{p^{k-s}}=1,\quad
b^{-1}ab=a^{1+p^s}
\right\rangle.
\]
Equivalently,
\[
P_{k,s}^{(p)}
\cong
\mathbb Z_{p^k}
\rtimes_{\,1+p^s}
\mathbb Z_{p^{k-s}}.
\]
The automorphism
\[
a\longmapsto a^{1+p^s}
\]
has order \(p^{k-s}\) on \(\langle a\rangle\), so the presentation
above defines the indicated semidirect product.  Moreover, the action
is a power automorphism of the form appearing in
\cref{prop:p-group-modular}.  Hence \(P_{k,s}^{(p)}\) is a modular
\(p\)-group and therefore is quasi-Hamiltonian.

The lower central series of \(P_{k,s}^{(p)}\) can be computed
directly.  Since
\[
[a^{p^j},b]
=
a^{-p^j}b^{-1}a^{p^j}b
=
a^{p^{j+s}},
\]
we obtain
\[
\gamma_2\bigl(P_{k,s}^{(p)}\bigr)
=
\langle a^{p^s}\rangle
\]
and, inductively,
\[
\gamma_{r+1}\bigl(P_{k,s}^{(p)}\bigr)
=
\langle a^{p^{rs}}\rangle.
\]
It follows that the group has nilpotency class $\left\lceil\frac{k}{s}\right\rceil$.
In particular, fixing an odd prime \(p\) and taking \(s=1\), the
nilpotency class becomes arbitrarily large as \(k\) increases.  Thus,
finite quasi-Hamiltonian groups do not have bounded nilpotency class.
For example,
\[
\mathbb Z_{81}\rtimes_{4}\mathbb Z_{27}
\]
is the case \(p=3\), \(k=4\), and \(s=1\), and has nilpotency class
\(4\).
\paragraph{Incomparability with groups of class at most two.}
The preceding family shows that quasi-Hamiltonian groups are not
contained in the class of groups of nilpotency class at most \(2\).
The converse containment also fails.  With our convention that
\(D_N\) has order \(2N\), the group
\[
D_4
=
\left\langle
r,s
\;\middle|\;
r^4=s^2=1,\quad srs=r^{-1}
\right\rangle
\]
has nilpotency class \(2\), but is not quasi-Hamiltonian.  Indeed, let
\[
H:=\langle s\rangle,
\qquad
K:=\langle rs\rangle.
\]
Then
\[
HK=\{1,s,rs,r^3\},
\qquad
KH=\{1,s,rs,r\},
\]
and therefore \(HK\neq KH\).  Hence \(H\) and \(K\) do not permute.
This proves that the classes of quasi-Hamiltonian groups and groups of
nilpotency class at most \(2\) are incomparable.

\section{The one-copy PGM and an efficient fixed-$t$ measurement}\label[appendix]{sec:app-PGM}
This appendix has two purposes.  We first analyze the exact one-copy
PGM for the ensemble
$$
\{\sigma_{d,t}:d\in\mathbb Z_N\}
$$
and derive a lower bound on its success probability.  We then give a
simpler block measurement that can be implemented efficiently and
achieves the same lower bound. 

Fix $t\in\{0,\ldots,k-1\}$ and write
\[
    m_t:=p^{k-t},\qquad L_t:=\frac{m_t}{p}=p^{k-t-1},\qquad
    \nu_t:=\mu^{p^t},
\]
and
\[
    M_t^{(b)}:=\sum_{i=0}^{b-1}\nu_t^i\pmod N,
    \qquad 0\leq b<m_t.
\]
Since $\mu^{p^k}=1$, we have $\nu_t^{m_t}=1$. Thus the order of
$\nu_t$ is a power of $p$ dividing $m_t$.

If $\nu_t=1$, then the subgroup
\[
    G_t:=\mathbb Z_N\rtimes \langle p^t\rangle
\]
is the Abelian group $\mathbb Z_N\times\mathbb Z_{m_t}$. Since
$H_{d,t}\leq G_t$, the fixed-$t$ instance can then be solved directly by
the Abelian HSP algorithm. Hence, throughout the remainder of this
appendix, assume that $\nu_t\neq 1$.

For fixed $t$, recall from Eq.~(22) that the relevant state is
\[
    \sigma_{d,t}
    =\frac1N\sum_{x\in\mathbb Z_N}|x\rangle\langle x|
      \otimes |\varphi_{x,d,t}\rangle\langle\varphi_{x,d,t}|,
\]
where
\[
    |\varphi_{x,d,t}\rangle
    =\frac1{\sqrt{m_t}}
      \sum_{b=0}^{m_t-1}\omega_N^{xM_t^{(b)}d}|b\rangle,
    \qquad \omega_N=e^{2\pi i/N}.
\]
As in \cite{Bacon2005optimalmeasurement}, it is convenient to define the PGM
for the ensemble $\{\sigma_{d,t}:d\in\mathbb Z_N\}$. The actual value of $d$ is one for which
$H_{d,t}$ is a subgroup of order $m_t$.

For $x,w\in\mathbb Z_N$, define
\[
    S_{w}^{x,t}:=
      \bigl\{b\in\{0,\ldots,m_t-1\}:
      xM_t^{(b)}=w\pmod N\bigr\},
    \qquad
    \eta_{w}^{x,t}:=|S_{w}^{x,t}|.
\]
Whenever $\eta_{w}^{x,t}>0$, let
\[
    |S_{w}^{x,t}\rangle
    :=\frac1{\sqrt{\eta_{w}^{x,t}}}
      \sum_{b\in S_{w}^{x,t}}|b\rangle.
\]
Then
\[
\sigma_{d,t}
=\frac1{Nm_t}\sum_{x\in\mathbb Z_N}
  \sum_{w,v\in\mathbb Z_N}
  \omega_N^{d(w-v)}
  \sqrt{\eta_{w}^{x,t}\eta_{v}^{x,t}}
  |x,S_{w}^{x,t}\rangle\langle x,S_{v}^{x,t}|.
\]

\begin{theorem}\label{thm:pgm-success-fixed-t}
For fixed $t$, the one-copy PGM for the ensemble
$\{\sigma_{d,t}:d\in\mathbb Z_N\}$ identifies $d$ with probability
\[
    P_{\mathrm{succ}}^{(t,1)}
    =\frac1{N^2m_t}
      \sum_{x\in\mathbb Z_N}
      \left(\sum_{w\in\mathbb Z_N}
      \sqrt{\eta_{w}^{x,t}}\right)^2.
\]
The probability is independent of $d$.
\end{theorem}

\begin{proof}
Let $\Sigma_t:=\sum_{d\in\mathbb Z_N}\sigma_{d,t}$. Summing over $d$
eliminates all terms with $w\neq v$, and therefore
\[
    \Sigma_t
    =\frac1{m_t}\sum_{x\in\mathbb Z_N}
      \sum_{w:\,\eta_{w}^{x,t}>0}
      \eta_{w}^{x,t}
      |x,S_{w}^{x,t}\rangle\langle x,S_{w}^{x,t}|.
\]
Taking the inverse square root on the support of $\Sigma_t$, the PGM
operator labelled by $c\in\mathbb Z_N$ is
\[
    E_c^{(t)}
    =\frac1N\sum_{x\in\mathbb Z_N}
      \sum_{w,v:\,\eta_{w}^{x,t}\eta_{v}^{x,t}>0}
      \omega_N^{c(w-v)}
      |x,S_{w}^{x,t}\rangle\langle x,S_{v}^{x,t}|.
\]
A direct block-by-block calculation gives
\[
    \operatorname{tr}\!\left(E_d^{(t)}\sigma_{d,t}\right)
    =\frac1{N^2m_t}
      \sum_{x\in\mathbb Z_N}
      \left(\sum_{w\in\mathbb Z_N}
      \sqrt{\eta_{w}^{x,t}}\right)^2.
\]
The expression does not depend on $d$.
\end{proof}

\begin{proposition}\label{prop:pgm-prime-block-bound}
Assume $\nu_t\neq 1$. Then the one-copy PGM satisfies
\[
    P_{\mathrm{succ}}^{(t,1)}
    \geq \frac{\varphi(N)p}{N^2}.
\]
No assumption that $\nu_t$ has order exactly $m_t$ is required.
\end{proposition}

\begin{proof}
Partition $\{0,\ldots,m_t-1\}$ into the $L_t=m_t/p$ consecutive
blocks
\[
    B_j:=\{pj+r:0\leq r<p\},
    \qquad 0\leq j<L_t.
\]
We first show that $b\mapsto M_t^{(b)}$ is injective on every block
$B_j$. Suppose that $0\leq r'<r<p$ and
\[
    M_t^{(pj+r)}=M_t^{(pj+r')}.
\]
Using
\[
    M_t^{(a+b)}=M_t^{(a)}+\nu_t^aM_t^{(b)},
\]
we obtain
\[
    0=M_t^{(pj+r)}-M_t^{(pj+r')}
      =\nu_t^{pj+r'}M_t^{(r-r')}.
\]
Because $\nu_t$ is a unit modulo $N$, this implies
$M_t^{(r-r')}=0$. Multiplying by $\nu_t-1$ and using
\[
    (\nu_t-1)M_t^{(s)}=\nu_t^s-1
\]
gives $\nu_t^{r-r'}=1$. The order of $\nu_t$ is a nontrivial power of
$p$, so $p$ divides $r-r'$. This is impossible because
$0<r-r'<p$. Hence the restriction of $b\mapsto M_t^{(b)}$ to each
$B_j$ is injective.

Now fix $x\in\mathbb Z_N^\times$. Multiplication by $x$ is a
permutation of $\mathbb Z_N$, so $b\mapsto xM_t^{(b)}$ is also
injective on every block. Consequently, a fixed value $w$ can have at
most one preimage in each of the $L_t$ blocks, and therefore
\[
    \eta_{w}^{x,t}\leq L_t
    \qquad\text{for all }w\in\mathbb Z_N.
\]
Also,
\[
    \sum_{w\in\mathbb Z_N}\eta_{w}^{x,t}=m_t.
\]
Since $0\leq\eta_{w}^{x,t}\leq L_t$, we have
\[
    \sqrt{\eta_{w}^{x,t}}
    \geq \frac{\eta_{w}^{x,t}}{\sqrt{L_t}}.
\]
It follows that
\[
    \left(\sum_w\sqrt{\eta_{w}^{x,t}}\right)^2
    \geq
    \left(\frac1{\sqrt{L_t}}\sum_w\eta_{w}^{x,t}\right)^2
    =\frac{m_t^2}{L_t}
    =m_tp.
\]
Keeping only the $\varphi(N)$ invertible values of $x$ in
Theorem~\ref{thm:pgm-success-fixed-t} gives
\[
    P_{\mathrm{succ}}^{(t,1)}
    \geq \frac1{N^2m_t}\,\varphi(N)m_tp
    =\frac{\varphi(N)p}{N^2}.
\]
\end{proof}
Since
$$
P_{\mathrm{succ}}^{(t,1)}
\geq
\frac{\phi(N)p}{N^2},
$$
the inverse success probability satisfies
$$
\frac{1}{P_{\mathrm{succ}}^{(t,1)}}
\leq
\frac{N}{p}\frac{N}{\phi(N)}.
$$
Using
$$
\frac{N}{\phi(N)}=O(\log\log N),
$$
the success probability is inverse-polynomially bounded whenever
$$
\frac{N}{p}\leq\operatorname{poly}(\log N).
$$

\subsection{An efficient block measurement}
 The PGM analysis above motivates a simpler measurement that can be
implemented efficiently.  Following the construction of
\cite[Section~5.2]{Bacon2005optimalmeasurement}, we now describe this
block measurement.  It is not necessary to implement the full
one-copy PGM: the block measurement below achieves the same lower
bound
\[
\frac{\phi(N)p}{N^2}
\]
on the success probability.  
  One first reduces the state
to a block containing only $p$ consecutive exponents.

Write $b=pj+r$, where $0\leq j<L_t$ and $0\leq r<p$, and apply the
reversible division map
\[
    |b\rangle\longmapsto |j,r\rangle.
\]
Measure $j$. Conditional on the outcome $j$, and ignoring a global
phase, the state $|\varphi_{x,d,t}\rangle$ becomes
\[
    |\chi_{x,j,d,t}\rangle
    =\frac1{\sqrt p}\sum_{r=0}^{p-1}
      \omega_N^{xM_t^{(r)}d_j}|r\rangle,
    \qquad
    d_j:=\nu_t^{pj}d.
\]
Indeed,
\[
    M_t^{(pj+r)}
    =M_t^{(pj)}+\nu_t^{pj}M_t^{(r)}.
\]
Since $\nu_t$ is a unit, recovering $d_j$ is equivalent to recovering
$d$, because
\[
    d=\nu_t^{-pj}d_j\pmod N.
\]

Measure the Fourier label $x$ and restart unless
$x\in\mathbb Z_N^\times$. This event has probability
$\varphi(N)/N$. Append an ancilla and compute
\[
    |r,0\rangle\longmapsto
    |r,z_r\rangle,
    \qquad z_r:=xM_t^{(r)}\pmod N.
\]
The values $z_r$, $0\leq r<p$, are distinct by the block-injectivity
argument above. The value of $r$ can be coherently recovered from
$z_r$: first compute
\[
    u_r:=1+(\nu_t-1)x^{-1}z_r=\nu_t^r\pmod N,
\]
and then compute the discrete logarithm of $u_r$ to the base $\nu_t$.
Because $0\leq r<p$ and the order of $\nu_t$ is a nontrivial power of
$p$, this logarithm determines $r$ uniquely. Using Shor's discrete
logarithm algorithm coherently, uncompute the $r$ register and obtain
\[
    \frac1{\sqrt p}\sum_{r=0}^{p-1}
      \omega_N^{z_rd_j}|z_r\rangle.
\]
The quantities $M_t^{(r)}$ can be computed in time polynomial in
$\log N+\log p$ by repeated squaring, exactly as in the
calculation of \cite{Bacon2005optimalmeasurement}.

Finally, apply the inverse Fourier transform over $\mathbb Z_N$. The
amplitude of the outcome $d_j$ is
\[
    \frac1{\sqrt{Np}}\sum_{r=0}^{p-1}1
    =\sqrt{\frac pN},
\]
so $d_j$ is observed with probability $p/N$. Taking account of the
probability $\varphi(N)/N$ that $x$ is invertible, one run succeeds
with probability at least
\[
    \frac{\varphi(N)p}{N^2}.
\]
After obtaining $d_j$, compute $d=\nu_t^{-pj}d_j$ and verify the
candidate by checking
\[
    f(0,0)=f(d,p^t).
\]

Therefore, if $N/p=\operatorname{poly}(\log N)$, repetition
\[
    O\!\left(\frac{N^2}{\varphi(N)p}
       \log\frac1\epsilon\right)
\]
times gives error at most $\epsilon$. Since
$N/\varphi(N)=O(\log\log N)$, the total running time is polynomial in
$\log N+k\log p$ and $\log(1/\epsilon)$.

\section{Reducing HSP from $G$ to HSP on Sylow factors}\label[appendix]{sec:Sylow-HSP}
The given Sylow decomposition allows us to reduce the HSP in $G$ to separate HSP instances in its Sylow factors. Let
\[
G=\prod_{p\mid |G|}G_p
\]
and let \(F:G\to S\) hide a subgroup \(H\leq G\). For each prime
$p\mid |G|$, define $H_p:=H\cap G_p$.
First, the restriction of \(F\) to \(G_p\) hides \(H_p\). Indeed,
for \(x,y\in G_p\),
\[
F(x)=F(y)
\quad\Longleftrightarrow\quad
x^{-1}y\in H.
\]
Since \(x^{-1}y\in G_p\), this is equivalent to $x^{-1}y\in H\cap G_p=H_p$.

We also have
\[
H=\prod_{p\mid |G|}H_p.
\]
To see this, let \(h\in H\), and write its unique Sylow
decomposition as
\[
h=\prod_{p\mid |G|}h_p,
\qquad h_p\in G_p.
\]
We claim that every component \(h_p\) belongs to \(H\).

Fix \(p\), and put
\[
M_p:=\frac{|G|}{|G_p|}.
\]
Since \(\gcd(M_p,|G_p|)=1\), there exists an integer \(a_p\) such
that
\[
a_pM_p\equiv1\pmod{|G_p|}.
\]
Set \(e_p=a_pM_p\). For every prime \(q\neq p\), the integer
\(|G_q|\) divides \(M_p\), and hence
\[
h_q^{e_p}=1.
\]
On the other hand,
\[
e_p\equiv1\pmod{|G_p|},
\]
so $h_p^{e_p}=h_p$ and it follows that $h^{e_p}=h_p$.
Since \(h\in H\) and \(H\) is closed under taking powers, we have
\(h_p\in H\). Therefore
\[
h_p\in H\cap G_p=H_p.
\]

Thus every \(h\in H\) belongs to \(\prod_pH_p\), proving
\[
H\subseteq\prod_pH_p.
\]
The reverse inclusion follows because every \(H_p\) is contained
in \(H\). Hence
\[
H=\prod_{p\mid |G|}H_p.
\]

Consequently, we solve the HSP separately in each Sylow factor
\(G_p\). If \(X_p\) is a generating set for \(H_p\), then
\[
\bigcup_{p\mid |G|}X_p
\]
is a generating set for \(H\).
\end{document}